\documentclass[journal]{IEEEtran}
\ifCLASSINFOpdf
\else
\fi
\usepackage[utf8]{inputenc} 
\usepackage[T1]{fontenc}
\usepackage{microtype}
\usepackage{textcomp}
\usepackage{xcolor}

\usepackage{amsmath, amssymb, amsfonts, amsthm, amstext, amsbsy}
\usepackage{mathtools}
\usepackage{mathrsfs}
\usepackage{bm}
\usepackage{dsfont}
\usepackage{stmaryrd}
\usepackage{xfrac}
\usepackage{nicefrac}
\usepackage{relsize}
\allowdisplaybreaks[1]

\usepackage{algorithm}
\usepackage{algpseudocode}

\usepackage{graphicx}
\usepackage{subcaption}
\usepackage{booktabs}
\usepackage{multirow}
\usepackage{array}
\usepackage{tabularx}
\usepackage{longtable}
\usepackage{makecell}
\usepackage{colortbl}
\usepackage{float}

\usepackage{graphicx}

\usepackage{setspace}
\usepackage{cuted}
\usepackage{multicol}
\usepackage{enumitem}
\usepackage{footmisc}
\usepackage{chngcntr}
\usepackage{tocloft}

\usepackage{tikz}
\usetikzlibrary{arrows.meta, positioning, calc}

\usepackage{cite}
\usepackage{doi}
\usepackage{bibunits}
\defaultbibliographystyle{IEEEtran}
\defaultbibliography{references}

\usepackage{hyperref}
\hypersetup{
    colorlinks=true,
    linkcolor={red!50!black},
    citecolor={blue!50!black},
    urlcolor={black}
}
\usepackage{orcidlink}

\definecolor{headercolor}{gray}{0.82}
\definecolor{zebracolor}{gray}{0.95}
\definecolor{lightgray}{gray}{0.93}

\def\BibTeX{{\rm B\kern-.05em{\sc i\kern-.025em b}\kern-.08em
    T\kern-.1667em\lower.7ex\hbox{E}\kern-.125emX}}

\def\D{\mathrm{D}} 

\newcommand{\orcid}[1]{%
    \href{https://orcid.org/#1}{\includegraphics[width=10pt]{orcid}}%
}

\newtheorem{theorem}{Theorem}

\newtheorem{definition}[theorem]{Definition}
\newtheorem{remark}[theorem]{Remark}
\newtheorem{proposition}[theorem]{Proposition}

\newtheorem{corollary}[theorem]{Corollary}

\newtheorem{assumption}{Assumption}

\def\markov{\hbox{$\--$}\kern-1.5pt\hbox{$\circ$}\kern-1.5pt\hbox{$\--$}}

\makeatletter
\newcommand*{\centernot}{\mathpalette\@centernot}
\def\@centernot#1#2{%
  \mathrel{%
    \rlap{%
      \settowidth\dimen@{$\m@th#1{#2}$}%
      \kern.5\dimen@
      \settowidth\dimen@{$\m@th#1=$}%
      \kern-.5\dimen@
      $\m@th#1\not$%
    }%
    {#2}%
  }%
}
\makeatother

\renewcommand\thesection{\arabic{section}}
\renewcommand\thesubsection{\arabic{section}.\arabic{subsection}}

\def\thesectiondis{\thesection.}
\def\thesubsectiondis{\thesectiondis\arabic{subsection}.}

\makeatletter
\def\@seccntformatinl#1{\csname the#1dis\endcsname\hskip 1em\relax}
\makeatother

\makeatletter
\newcommand*{\rom}[1]{\expandafter\@slowromancap\romannumeral #1@}
\makeatother

\usepackage{cleveref}

\makeatletter

\let\lpgrand@original@section\section
\let\lpgrand@original@subsection\subsection

\NewDocumentCommand{\lpgrand@section}{s o m}{%
  \IfBooleanTF{#1}{%
    \lpgrand@original@section*{#3}%
  }{%
    \IfNoValueTF{#2}{%
      \lpgrand@original@section{#3}%
      \addcontentsline{atoc}{section}{%
        \protect\numberline{\thesection}#3}%
    }{%
      \lpgrand@original@section[#2]{#3}%
      \addcontentsline{atoc}{section}{%
        \protect\numberline{\thesection}#2}%
    }%
  }%
}

\NewDocumentCommand{\lpgrand@subsection}{s o m}{%
  \IfBooleanTF{#1}{%
    \lpgrand@original@subsection*{#3}%
  }{%
    \IfNoValueTF{#2}{%
      \lpgrand@original@subsection{#3}%
      \addcontentsline{atoc}{subsection}{%
        \protect\numberline{\thesubsection}#3}%
    }{%
      \lpgrand@original@subsection[#2]{#3}%
      \addcontentsline{atoc}{subsection}{%
        \protect\numberline{\thesubsection}#2}%
    }%
  }%
}

\newcommand{\AppendixOnlyTOC}{%
  \lpgrand@original@section*{Appendix Contents}%
  \begingroup
    \footnotesize
    \setcounter{tocdepth}{2}%
    \setlength{\parskip}{0pt}%
    \@starttoc{atoc}%
  \endgroup
  \let\section\lpgrand@section
  \let\subsection\lpgrand@subsection
}

\makeatother

\begin{document}

\begin{bibunit}

\title{On the Capacity of Distinguishable Synthetic Identity Generation under Face Verification}

\author{Behrooz~Razeghi\textsuperscript{\orcidlink{0000-0001-9568-4166}},~\IEEEmembership{Senior~Member,~IEEE}
\thanks{B. Razeghi is with the School of Engineering and Applied Sciences, Harvard University, Cambridge,
MA 02134 USA (e-mail: behroozrazeghi@seas.harvard.edu).}
\thanks{\urlstyle{sf}The official source code is available at: \url{https://github.com/BehroozRazeghi/synthetic-identity-capacity}.}%
%
}

%



\maketitle

\begin{abstract}
Synthetic face generators can produce many nominal identities, but nominal count does not determine how many are jointly distinguishable under a specified verification rule. We define finite-dimensional capacity as the
supremum of codebook sizes over distinct latent identity codes whose induced identity-conditional embedding distributions satisfy per-identity genuine acceptance and pairwise impostor non-match constraints. For deterministic view-invariant pipelines, fixed-code capacity equals the spherical-code cardinality over the realizable embedding set and reduces to the classical spherical-code cardinality when every sphere direction is realizable. For stochastic identity-conditional embedding distributions concentrated with probability at least $1-\eta$ in spherical caps of angular radius $\rho$, we derive a sufficient center-separation condition, spherical-code capacity lower bounds under full angular expressivity, and positive asymptotic lower-bound exponents for dimension-indexed pipeline families. We also derive prior-constrained random-code lower bounds from pairwise center-separation failure probabilities. When each identity-conditional embedding distribution has support equal to a spherical cap of angular radius $\rho$, we derive necessary zero-error geometric conditions and, for $2\rho<\arccos(\tau)$ under full $\rho$-cap angular expressivity, show that the restricted zero-error capacity equals the classical spherical-code cardinality at minimum angle $\arccos(\tau)+2\rho$. For finite repeated-view samples, a maximum clique in the resulting compatibility graph identifies the largest sampled subset satisfying all empirical genuine and pairwise impostor constraints. We evaluate this sample-restricted quantity on a deterministically selected DigiFace-1M subset under three fixed recognizers with identity-disjoint in-domain threshold calibration.
\end{abstract}

\begin{IEEEkeywords}
Face verification, synthetic identity generation, identity capacity, hyperspherical embeddings, spherical codes.
\end{IEEEkeywords}

%
\IEEEpeerreviewmaketitle

\section{Introduction}
\label{sec:introduction}

\IEEEPARstart{T}{he} number of nominal identities produced by a synthetic face generator does not determine how many are jointly distinguishable under a specified verification rule. Distinguishability requires simultaneous control of same-identity match probabilities and distinct-identity non-match probabilities. We formalize this requirement as a finite-dimensional identity capacity for a fixed generator--recognizer pipeline subject to per-identity genuine acceptance and pairwise impostor non-match probability constraints.

Representation-learning work has studied view agreement, representation spread, and information-theoretic criteria \cite{galvez2023role, oord2018representation, bachman2019learning, wang2020understanding, wu2020mutual, shwartz2023information, yerxa2023learning, isik2023information, schaeffer2024towards, razeghi2026deep}. Aggregate measures of view agreement or representation spread, information-theoretic criteria, average similarity, and population-averaged verification rates do not, by themselves, imply the required lower bounds on genuine acceptance probability for every identity and impostor non-match probability for every distinct identity pair. Recent synthetic-face work studies image realism, identity consistency, diversity, and synthetic-data utility across GAN, diffusion, and hybrid generation pipelines \cite{qiu2021synface, colbois2021use, kim2023dcface, melzi2023gandiffface, boutros2023synthetic, sun2024cemiface, shahreza2025hyperface, rahimi2024synthetic, caldeira2025negfacediff, kim2025vigface}. Related capacity and identity-provisioning formulations use manifold packing, hyperspherical-cap geometry, or explicit geometric separation~\cite{gong2017capacity,boddeti2023biometric,ji2026bip}. The capacity definition instead requires the identity-conditional embedding distributions induced by a codebook of pairwise-distinct latent identity codes to satisfy per-identity genuine acceptance and pairwise impostor non-match probability lower bounds under a fixed generator--recognizer pipeline and verification rule.

Spherical codes impose a minimum pairwise angular separation among deterministic embedding directions, whereas the capacity definition imposes lower bounds on same-identity match and distinct-identity non-match probabilities. For deterministic view-invariant pipelines, each identity-conditional embedding distribution is a point mass, and the capacity over deterministically selected latent identity codes equals the spherical-code cardinality over the realizable embedding directions. If, for every sphere direction, some latent identity code induces that embedding almost surely, the capacity equals the classical spherical-code cardinality \cite{rankin1955closest,conway2013sphere}. For the centered class, each identity-conditional distribution assigns at least a prescribed probability to a spherical cap of prescribed angular radius about an identity center. If any two points in the same cap satisfy the match rule and any two points from distinct identity caps satisfy the non-match rule, then, by independence, each required genuine acceptance probability and each required impostor non-match probability is at least the square of the prescribed lower bound on the probability of falling in the cap. If every finite set of pairwise-distinct prescribed sphere directions can be realized as centers of such distributions by pairwise-distinct latent codes, the cardinality of a classical spherical code whose minimum angle guarantees the cross-cap non-match condition gives a lower bound on the capacity restricted to the
centered class. If each identity-conditional distribution instead has support equal to its spherical cap, zero-error admissibility requires the angular diameter of the cap not to exceed the angular threshold induced by the verification threshold and requires pairwise center separation of at least the sum of that threshold and the cap diameter. When the cap diameter is strictly smaller than the angular threshold and every finite set of pairwise-distinct prescribed cap centers can be realized by pairwise-distinct latent codes whose induced distributions have those caps as their supports, the restricted zero-error capacity equals the classical spherical-code cardinality with minimum angle equal to the required pairwise center separation.

For finite repeated-view samples, identities satisfying the empirical genuine acceptance requirement form the vertices of a compatibility graph, and two vertices are joined by an edge when the corresponding identity pair satisfies the empirical impostor non-match requirement. Any maximum clique corresponds to a sampled identity subset of maximum cardinality satisfying every empirical genuine acceptance and pairwise impostor non-match requirement. Without additional assumptions on sampling and latent-space coverage, this finite-sample cardinality is neither an estimator nor a confidence bound for the supremum of sizes of latent identity codebooks whose induced distributions satisfy the defining probability constraints. The image-level study evaluates this finite-sample quantity under fixed recognizers and a specified protocol for verification-threshold calibration.


\paragraph*{Contributions}

Our contributions are summarized below.

\begin{enumerate}[leftmargin=*]
\item We define fixed-code operational identity capacity for a fixed generator--recognizer pipeline through per-identity genuine acceptance and pairwise impostor non-match probability constraints. We also define a prior-constrained random-code identity capacity requiring pairwise distinctness and the same constraints to hold jointly with probability at least $1-\delta$ over independently sampled latent identity codes.
\item For deterministic view-invariant pipelines, we characterize the fixed-code capacity for every $\varepsilon_{\rm in},\varepsilon_{\rm out}\in[0,1)$ as the spherical-code cardinality over the realizable embedding set, with equality to the classical spherical-code cardinality when every sphere direction is deterministically realizable. For $(\rho,\eta)$-centered identity-conditional embedding distributions, we establish sufficient cap-radius and center-separation conditions and, under full $(\rho,\eta)$-angular expressivity, a spherical-code lower bound on the $(\rho,\eta)$-restricted capacity. For dimension-indexed pipeline families that are fully $(\rho,\eta)$-angularly expressive for all sufficiently large dimensions, $2\rho<\arccos(\tau)$ and $\arccos(\tau)+2\rho<\pi/2$ yield an asymptotic rate lower bound.
\item Under the centered random-code model and $2\rho<\arccos(\tau)$, we derive lower bounds on prior-constrained random-code capacity from the pairwise center-separation failure probability. When each identity-conditional embedding distribution has support equal to a spherical cap of angular radius $\rho$, we derive necessary zero-error geometric conditions and show that, at $2\rho=\arccos(\tau)$, support alone does not determine zero-error genuine acceptance. When $2\rho<\arccos(\tau)$ and full $\rho$-cap angular expressivity holds, the full-cap restricted zero-error capacity equals the classical spherical-code cardinality at minimum angle $\arccos(\tau)+2\rho$.
\item For finite repeated-view samples, we construct a compatibility graph in which a maximum clique corresponds to a sampled identity subset of maximum cardinality satisfying all empirical genuine acceptance and pairwise impostor non-match requirements. We evaluate this finite-sample cardinality on a deterministically selected DigiFace-1M subset under three fixed recognizers with separate identity-disjoint in-domain threshold calibration for each recognizer.
\end{enumerate}

\paragraph*{Related Work}

Capacity has previously been studied for face representations and generative face models. Gong \textit{et al.}~\cite{gong2017capacity} formulate face-representation capacity as a packing problem on a low-dimensional manifold embedded in a deep representation space while accounting for epistemic and aleatoric representational uncertainty. Boddeti \textit{et al.}~\cite{boddeti2023biometric} estimate an upper bound on the biometric capacity of generative face models in a fixed hyperspherical face-recognition space by approximating the population and identity-specific manifolds with hyperspherical caps and taking the ratio of their surface areas at a specified matcher operating point. For the unconditional generators in their study, the identity-specific angular span cannot be estimated directly from repeated samples of a fixed generated identity; for the class-conditional DCFace model, repeated identity-conditioned samples permit direct estimation. More recent methods explicitly construct virtual identities in face-recognition embedding spaces. Kim \textit{et al.}~\cite{kim2025vigface} allocate virtual prototypes in a face-recognition embedding space and condition a diffusion model on these prototypes to synthesize identity-consistent faces. Ji \textit{et al.}~\cite{ji2026bip} formulate biometric identity provisioning as a constrained packing problem in which virtual identities must avoid an enrolled real-person gallery, remain mutually separated, and be realizable as face images.

Our formulation defines capacity directly through verification-event probabilities under a fixed verification rule. For a fixed generator--recognizer pipeline, we consider pairwise-distinct latent identity codes whose induced identity-conditional embedding distributions satisfy prescribed per-identity genuine acceptance and pairwise impostor non-match probability constraints. The fixed-code capacity is the supremum of codebook sizes satisfying these constraints. This differs from estimating capacity through a population-to-identity cap-area ratio or from allocating identity representations according to geometric separation constraints. The operational definition itself does not require a hyperspherical-cap population model, a representational-noise model, or angular expressivity. Angular expressivity is assumed only to guarantee that prescribed spherical-code centers can be realized by pairwise-distinct latent codes with the required identity-conditional embedding distributions.

\paragraph*{Organization}
Section~\ref{sec:problem_formulation} defines the fixed-code and prior-constrained random-code capacities and the finite-sample operational formulation. Section~\ref{sec:geometric_identity_capacity} analyzes deterministic view-invariant and centered stochastic fixed-code regimes, finite-dimensional spherical-code bounds, and dimension-indexed asymptotic achievability bounds. Section~\ref{sec:random_code_capacity} analyzes prior-constrained random coding under the centered model. Section~\ref{sec:image_level_evaluation} presents the finite-sample DigiFace-1M evaluation and the effective-dimension sensitivity analysis. Complete proofs, auxiliary theoretical results, the full-cap-support converse, spherical-code numerical evaluations, and the image-level evaluation protocol and reproducibility records are provided in the Appendices.

\section{System Model and Problem Formulation}
\label{sec:problem_formulation}

\begin{figure*}
    \centering
    \includegraphics[width=0.79\linewidth]{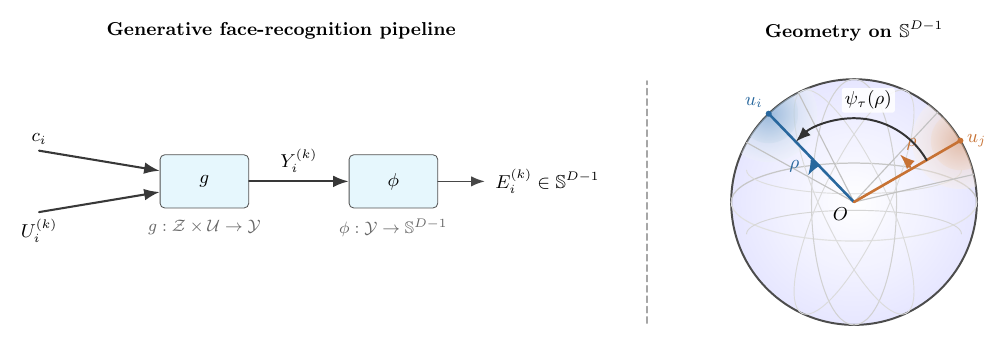}
    \caption{Generative face-recognition pipeline and centered identity-embedding geometry.}
    \label{fig:identity_capacity_blockdiagram}
\end{figure*}

\subsection{Generative Face-Recognition Pipeline}
\label{subsec:gen_fr_pipeline}

We model a generative face-recognition pipeline as a latent-code-based generator followed by a deterministic normalized embedding map. To represent repeated observations of a generator-defined identity, we allow nuisance-varying realizations associated with pose, illumination, expression, rendering stochasticity, and other within-identity variation. The resulting model is specified through identity-conditional embedding distributions.

Let $\mathcal Z$ denote the latent space, $\mathcal U$ the nuisance space, and $\mathcal Y$ the face-image space. Let $D\ge2$, and for $n\in\mathbb N$, write $[n]\triangleq\{1,\ldots,n\}$. We equip these spaces with sigma-algebras, assume that $\mathcal Z$ is a standard Borel space, and assume that $(c,u)\mapsto\phi(g(c,u))$ is measurable. The standard Borel assumption ensures that the pairwise-distinct-code events used in the random-code formulation are measurable. For each identity $i\in[M]$, let $c_i\in\mathcal Z$ denote its latent identity code. Let $\{U_i^{(k)}\}_{i\in[M],\,k\in\mathbb N} \stackrel{\mathrm{i.i.d.}}{\sim}P_U$ represent within-identity nuisance variation. For each identity $i$ and view index $k$, define
\begin{equation}
Y_i^{(k)} \triangleq g(c_i,U_i^{(k)}),
\qquad E_i^{(k)} \triangleq \phi(Y_i^{(k)})\in\mathbb S^{D-1},
\label{eq:multiview_pipeline_main}
\end{equation}
where $g:\mathcal Z\times\mathcal U\to\mathcal Y$ is the generator, $\phi:\mathcal Y\to\mathbb S^{D-1}$ is the deterministic embedding map, and $\mathbb S^{D-1} \triangleq \{x\in\mathbb R^D:\|x\|_2=1\}$. For a fixed codebook $(c_1,\ldots,c_M)$, each sequence $(E_i^{(k)})_{k\ge1}$ is i.i.d., and the sequences corresponding to distinct identities are independent.

For each $c\in\mathcal Z$, define the induced embedding distribution $P_c$ on $\mathbb S^{D-1}$ by
\begin{equation}
P_c(A) \triangleq \Pr \left[\phi(g(c,U))\in A\right], \;
A\in\mathcal B(\mathbb S^{D-1}), \;
U\sim P_U,
\label{eq:Pc_def}
\end{equation}
where $\mathcal B(\mathbb S^{D-1})$ denotes the Borel sigma-algebra on $\mathbb S^{D-1}$. The measurability assumption above implies that $P_c$ is a probability measure for each $c$ and that $c\mapsto P_c(A)$ is measurable for every $A\in\mathcal B(\mathbb S^{D-1})$. Under \eqref{eq:multiview_pipeline_main}, $E_i^{(k)}\sim P_{c_i}$, $i\in[M],\ k\in\mathbb N$.
For a fixed codebook $(c_1,\ldots,c_M)$, we write $P_i\triangleq P_{c_i}$, $i\in[M]$, and refer to $\{P_i\}_{i=1}^M$ as the induced family of identity-conditional embedding distributions. Similarity is measured by the cosine score
\begin{equation}
S(e,e') \triangleq \langle e,e'\rangle,
\qquad e,e'\in\mathbb S^{D-1},
\label{eq:cosine_similarity_main}
\end{equation}
with verification threshold $\tau\in[0,1)$. We also use the angular distance
\begin{equation}
\angle(e,e') \triangleq \arccos\langle e,e'\rangle,
\qquad e,e'\in\mathbb S^{D-1}.
\label{eq:angle_definition_main}
\end{equation}
We adopt the convention that the verifier declares a match if $\langle e,e'\rangle>\tau$ and a non-match if $\langle e,e'\rangle\le\tau$. Equivalently, a match occurs when $\angle(e,e')<\arccos(\tau)$, and a non-match occurs when $\angle(e,e')\ge\arccos(\tau)$. The explicit tie convention is required for statements involving probability mass at the verification threshold.

We first consider a fixed-code formulation, in which the latent identity codes are selected deterministically. This corresponds to a deterministic codebook in information theory. We then consider a prior-constrained random-code formulation in which the identity codes are sampled independently from a prescribed probability measure on $\mathcal Z$. The random-code formulation applies when the sampled variable indexes a repeatable identity-conditional generation mechanism.

\subsection{Fixed-Code Operational Identity Capacity}
\label{subsec:fixed-code-operational_capacity}

Verification is specified through two tail constraints. Genuine pairs, formed from two independent views of the same identity, are required to have score strictly above the verification threshold with a prescribed minimum probability. Impostor pairs, formed from independent views of distinct identities, are required to have score at most the threshold with a prescribed minimum probability.

\begin{definition}[Fixed-Code Operational Identity Capacity]
\label{def:operational_capacity_main}
Fix $(g,\phi)$, embedding dimension $D\ge2$, threshold $\tau\in[0,1)$, and tolerances $\varepsilon_{\rm in},\varepsilon_{\rm out}\in[0,1)$. The nuisance distribution $P_U$ is fixed as part of the pipeline specification and is suppressed from the capacity notation. A family of identity-conditional embedding distributions $\{P_i\}_{i=1}^M$ is called \textit{$(\tau,\varepsilon_{\rm in},\varepsilon_{\rm out})$-admissible} if
\begin{subequations}
\begin{align}
\Pr \left[\left\langle E_i^{(1)},E_i^{(2)} \right\rangle > \tau \right] &\ge 1-\varepsilon_{\rm in},
&&\forall i\in[M],
\label{eq:admissible_intra_main}
\\
\Pr \left[\left\langle E_i^{(1)},E_j^{(1)} \right\rangle \le \tau \right] &\ge 1-\varepsilon_{\rm out},
&&\forall i\ne j,
\label{eq:admissible_inter_main}
\end{align}
\end{subequations}
where $E_i^{(1)},E_i^{(2)}\stackrel{\mathrm{i.i.d.}}{\sim}P_i$, and, for $i\ne j$, $E_j^{(1)}\sim P_j$ is independent of $E_i^{(1)}$.

The fixed-code operational identity capacity of $(g,\phi)$ is
\begin{multline}
\!\!\!
C_D(\tau,\varepsilon_{\rm in},\varepsilon_{\rm out};g,\phi) \triangleq
\sup\Bigl( \{0\} \cup \Bigl\{ M\in\mathbb N:\\ \exists\,c_1,\ldots,c_M\in\mathcal Z,  c_i\ne c_j \,
\forall i,j\in[M],\ i\ne j, \\
\quad\text{s.t.}\quad \{P_{c_i}\}_{i=1}^M
\text{ is }
(\tau,\varepsilon_{\rm in},\varepsilon_{\rm out})\text{-admissible} \Bigr\}\Bigr).
\label{eq:capacity_definition_main}
\end{multline}
\end{definition}

The supremum in \eqref{eq:capacity_definition_main} is taken in $\mathbb N_0\cup\{+\infty\}$; the explicit element $0$ ensures $C_D=0$ when no single-identity codebook satisfies the genuine constraint. Suppose two distinct codes in a codebook induce the same distribution $P$, and define
\[
p\triangleq \Pr_{E,E'\stackrel{\mathrm{i.i.d.}}{\sim}P} \left[\langle E,E'\rangle>\tau\right].
\]
The genuine acceptance probability for either identity is $p$, whereas their cross-identity non-match probability is $1-p$. Admissibility would therefore require $p\ge1-\varepsilon_{\rm in}$, $p\le\varepsilon_{\rm out}$. These inequalities cannot both hold when $\varepsilon_{\rm in}+\varepsilon_{\rm out}<1$. We refer to this condition as the nondegenerate regime.

Definition~\ref{def:operational_capacity_main} is finite-dimensional. A fixed pipeline $(g,\phi)$ has one embedding dimension. To define asymptotic growth with $D$, let $\mathcal D\subseteq\{2,3,\ldots\}$ be unbounded and let $\mathbf P \triangleq \{(g_D,\phi_D):D\in\mathcal D\}$ be a dimension-indexed family of pipelines of the form specified above, with $\phi_D:\mathcal Y_D\to\mathbb S^{D-1}$. Define
\begin{equation}
R(\tau,\varepsilon_{\rm in},\varepsilon_{\rm out};\mathbf P) \triangleq
\limsup_{\substack{D\to\infty\\D\in\mathcal D}} \frac{1}{D} \log
C_D(\tau,\varepsilon_{\rm in},\varepsilon_{\rm out};g_D,\phi_D),
\label{eq:rate_definition_main}
\end{equation}
where $\log$ denotes the natural logarithm. The rate is measured in nats per embedding dimension, with the extended-real conventions $\log 0=-\infty$ and $\log(+\infty)=+\infty$.

\begin{remark}[Role of embedding dimension]
\label{rem:embedding_dimension_role}
For a fixed trained recognizer, $C_D$ is evaluated at its native output dimension. The rate in \eqref{eq:rate_definition_main} instead concerns a dimension-indexed family of pipelines. Projecting the embeddings of one trained recognizer to lower-dimensional subspaces provides an effective-dimension sensitivity analysis, but does not realize the dimension-indexed family required by the asymptotic definition.
\end{remark}

\begin{remark}[Relation to standard verification metrics]
\label{rem:tmr_fmr_relation}
At threshold $\tau$, $\Pr \left[ \langle E_i^{(1)},E_i^{(2)}\rangle>\tau \right]$ is the genuine acceptance probability for identity $i$, whereas $\Pr \left[ \langle E_i^{(1)},E_j^{(1)}\rangle>\tau \right]$,  $i\ne j$, is the impostor match probability for pair $(i,j)$. Thus $\varepsilon_{\rm in}$ upper-bounds the genuine rejection probability for every identity, and $\varepsilon_{\rm out}$ upper-bounds the impostor match probability for every distinct identity pair. These are identity-wise and pairwise constraints rather than population-averaged verification rates.
\end{remark}

\begin{definition}[Protocol-calibrated fixed-code identity capacity]
\label{def:protocol_calibrated_capacity}
Let $\mathcal A$ denote a fully specified operating-point calibration protocol, including its calibration sample or distribution, preprocessing, score convention, and threshold-selection rule. Fix a target FMR
$\alpha\in[0,1)$ for which $\mathcal A$ returns a threshold for recognizer $\phi$, and denote that threshold by $\tau_{\mathcal A,\alpha}(\phi)\in[0,1)$. Define
\begin{equation}
C_D^{\mathcal A} (\alpha,\varepsilon_{\rm in},\varepsilon_{\rm out};g,\phi)
\triangleq C_D \left( \tau_{\mathcal A,\alpha}(\phi), \varepsilon_{\rm in}, \varepsilon_{\rm out};
g,\phi \right).
\label{eq:protocol_calibrated_capacity}
\end{equation}
For fixed $\mathcal A$ and $\phi$, the dependence on $\alpha$ enters through $\tau_{\mathcal A,\alpha}(\phi)$. The value of $\alpha$ alone does not determine a common numerical threshold across recognizers or calibration distributions. The calibration target $\alpha$ determines the threshold through $\mathcal A$, whereas $\varepsilon_{\rm out}$ bounds the impostor match probability for every distinct identity pair at that threshold. The definition does not require $\alpha=\varepsilon_{\rm out}$.
\end{definition}

\subsection{Prior-Constrained Random-Code Identity Capacity}

The fixed-code capacity in Definition~\ref{def:operational_capacity_main} is an existence statement, i.e., it optimizes over deterministic choices of latent codes and takes the supremum over deterministic codebook sizes for which the induced embedding distributions satisfy the prescribed verification constraints. In pipelines with an explicit identity-indexing latent or conditioning variable, candidate identity codes need not be selected by explicit codebook design; they may instead be sampled independently from a specified prior. This leads to a random-code analogue of the fixed-code capacity. This random-code formulation requires a latent variable that indexes a repeatable identity-conditional distribution; sampling independent image latents without such an identity index does not satisfy this requirement.

Let $P_C$ be a probability measure on $\mathcal Z$, and let $C_1,\ldots,C_M \stackrel{\mathrm{i.i.d.}}{\sim} P_C$. Define the distinct-code event $\mathcal D_M\triangleq\{C_i\ne C_j\text{ for all }i\ne j\}$. A sampled codebook is counted only on $\mathcal D_M$, consistently with the fixed-code definition. For each sampled code $C_i$, the induced embedding distribution is $P_{C_i}$ in the sense of \eqref{eq:Pc_def}. For compactness, let $\mathcal A_M(a,b)$ denote the event that $\{P_{C_i}\}_{i=1}^M$ is $(\tau,a,b)$-admissible. The resulting random family $\{P_{C_i}\}_{i=1}^M$ is admissible or not depending on the realized sample $(C_1,\ldots,C_M)$. We therefore define capacity in this setting as the supremum of codebook sizes for which pairwise distinctness and admissibility hold jointly with probability at least $1-\delta$ over the draw of the latent codes.

\begin{definition}[Prior-constrained random-code identity capacity]
\label{def:random_capacity_main}
Fix $(g,\phi)$, embedding dimension $D\ge 2$, threshold $\tau\in[0,1)$, tolerances $\varepsilon_{\rm in},\varepsilon_{\rm out}\in[0,1)$, a probability measure $P_C$ on $\mathcal Z$, and a failure-probability tolerance $\delta\in(0,1)$. Let $C_1,\ldots,C_M \stackrel{\mathrm{i.i.d.}}{\sim} P_C$, and let $P_{C_i}$ denote the identity-conditional embedding distribution induced by the random code $C_i$. The $(1-\delta)$ prior-constrained random-code identity capacity is defined by
\begin{multline}
\!\!\!\!\!\!
C_{D,\delta}^{\mathrm{rnd}} (\tau,\varepsilon_{\rm in},\varepsilon_{\rm out};g,\phi,P_C)
\triangleq \sup\Bigl(\{0\}\cup\Bigl\{ M\in\mathbb N:\;\;\; \\ \;\;\;
\Pr\bigl[ \mathcal D_M\cap \mathcal A_M(\varepsilon_{\rm in},\varepsilon_{\rm out}) \bigr]
\ge 1-\delta \Bigr\}\Bigr),
\label{eq:random_capacity_def_main}
\end{multline}
where the probability is with respect to $(C_1,\ldots,C_M)\sim P_C^{\otimes M}$.
\end{definition}

The kernel measurability assumed in Section~\ref{subsec:gen_fr_pipeline} implies that $\mathcal A_M(a,b)$ is measurable, while the standard Borel assumption on $\mathcal Z$ makes $\mathcal D_M$ measurable. Hence the event in \eqref{eq:random_capacity_def_main} is well defined.

\begin{proposition}[Random coding cannot exceed fixed-code capacity]
\label{prop:random_vs_deterministic_capacity}
For every $\delta\in(0,1)$,
\begin{equation}
C_{D,\delta}^{\mathrm{rnd}} (\tau,\varepsilon_{\rm in},\varepsilon_{\rm out};g,\phi,P_C)
\le C_D(\tau,\varepsilon_{\rm in},\varepsilon_{\rm out};g,\phi).
\label{eq:random_vs_deterministic_capacity}
\end{equation}
\end{proposition}

\begin{proof}
A proof is provided in the Appendix.
\end{proof}

\subsection{Finite-Sample Evaluation on a Fixed Pipeline}
\label{subsec:finite_sample_evaluation}

The operational constraints can be evaluated on finite repeated-view samples from a fixed synthetic-face pipeline. For candidate identities $i\in[N]$ in the evaluated sample, let $K_i\ge2$ and let $E_i^{(1)},\ldots,E_i^{(K_i)}$ be unit embeddings produced by a fixed recognizer. Define the empirical genuine acceptance rate
\begin{equation}
\widehat p_i^{\rm gen}(\tau) =\frac{1}{\binom{K_i}{2}}
\sum_{1\le a<b\le K_i} \mathbf 1 \left\{\left\langle E_i^{(a)},E_i^{(b)}\right\rangle>\tau\right\},
\label{eq:empirical_genuine_rate}
\end{equation}
and, for $i\ne j$, the empirical impostor non-match rate
\begin{equation}
\widehat p_{ij}^{\rm imp}(\tau) =\frac{1}{K_iK_j}
\sum_{a=1}^{K_i}\sum_{b=1}^{K_j} \mathbf 1 \left\{\left\langle E_i^{(a)},E_j^{(b)}\right\rangle\le\tau\right\}.
\label{eq:empirical_impostor_rate}
\end{equation}
Retain identities satisfying $\widehat p_i^{\rm gen}(\tau)\ge1-\varepsilon_{\rm in}$ as vertices, and let $N_{\rm G}$ denote their number. Connect two retained identities $i$ and $j$ when $\widehat p_{ij}^{\rm imp}(\tau)\ge1-\varepsilon_{\rm out}$. The maximum-clique cardinality of this compatibility graph is the largest subset of candidate identities in the evaluated sample that satisfies both constraints under these empirical rates. We denote it by $\widehat C_{\mathcal S}(\tau,\varepsilon_{\rm in},\varepsilon_{\rm out})$, where $\mathcal S$ is the evaluated finite sample.

If the repeated views are sampled independently from the evaluation nuisance distribution, \eqref{eq:empirical_genuine_rate} is a U-statistic of order two and \eqref{eq:empirical_impostor_rate} is a two-sample U-statistic of order $(1,1)$; the summands are nevertheless dependent when they share an image. For a deterministic render grid, the same formulas are descriptive complete-pair rates rather than i.i.d.-sampling estimators. Accordingly, $\widehat C_{\mathcal S}$ is a finite-sample cardinality. Without additional sampling and latent-space coverage assumptions, it is neither an estimator nor a confidence bound for the latent-space supremum $C_D$. The implementation records whether the returned maximum-clique cardinality is exact; when an exact solution is not obtained, the returned cardinality is explicitly identified as a lower bound.

\section{Geometric Analysis of Fixed-Code Identity Capacity}
\label{sec:geometric_identity_capacity}

\subsection{Deterministic View-Invariant Fixed-Code Capacity}
\label{subsec:deterministic_capacity}

We first consider a view-invariant pipeline, in which the embedding associated with each latent identity code is constant almost surely with respect to nuisance variation. For a general pipeline, define the set of deterministically realizable embedding directions by
\begin{equation}
\!\!\!\!
\mathcal{V}_{g,\phi} \triangleq
\left\{ u\in\mathbb{S}^{D-1} : \exists\, c\in\mathcal{Z} \;\; \text{s.t.}\;\;
\phi(g(c,U))=u \;\; \text{a.s.} \right\}.
\label{eq:realizable_set_main}
\end{equation}
For a subset $\mathcal V\subseteq \mathbb S^{D-1}$ and an angle $\psi\in(0,\pi]$, define
\begin{multline}
A(\mathcal V,\psi) \triangleq \max\Bigl(\{0\}\cup\Bigl\{
M\in\mathbb N:\exists\, u_1,\ldots,u_M\in\mathcal V \\
\text{such that }
\angle(u_i,u_j)\ge \psi,\ \forall i\neq j \Bigr\}\Bigr).
\label{eq:restricted_spherical_code_main}
\end{multline}
In particular, for the full sphere $\mathbb S^{D-1}$, define
\begin{equation}
A_D(\psi)\triangleq A(\mathbb S^{D-1},\psi),
\label{eq:spherical_code_def_main}
\end{equation}
which is the classical spherical-code packing number in dimension $D$; see, e.g., \cite{rankin1955closest, conway2013sphere}.

\begin{remark}
The packing number may also be parameterized by an inner-product threshold. For $t\in[-1,1)$, define $A_{\mathrm{thr}}(\mathcal V,t) \triangleq \max\Bigl(\{0\}\cup\Bigl\{ M\in\mathbb N:\exists\,u_1,\ldots,u_M\in\mathcal V \text{s.t. }\langle u_i,u_j\rangle\le t,\ \forall i\neq j \Bigr\}\Bigr)$. Since $\angle(u_i,u_j)\ge \psi$ is equivalent to $\langle u_i,u_j\rangle\le \cos\psi$, one has $A(\mathcal V,\psi)=A_{\mathrm{thr}}(\mathcal V,\cos\psi)$ for every $\psi\in(0,\pi]$. In particular, $A_{\mathrm{thr}}(\mathcal V,\tau)=A(\mathcal V,\arccos\tau)$ for $\tau\in[0,1)$. We use the angular parameterization $A(\mathcal V,\psi)$.
\end{remark}

\begin{proposition}[Deterministic view-invariant fixed-code capacity]
\label{prop:deterministic_capacity_main}
Assume that, for every $c\in\mathcal Z$, there exists $v(c)\in\mathbb S^{D-1}$ such that
\begin{equation}
\phi(g(c,U))=v(c)\qquad\mathrm{a.s.}
\label{eq:pipeline_view_invariance}
\end{equation}
Then $\mathcal V_{g,\phi}=\{v(c):c\in\mathcal Z\}$. Let
\begin{equation}
\psi_\tau\triangleq\arccos(\tau).
\end{equation}
For every $\varepsilon_{\rm in},\varepsilon_{\rm out}\in[0,1)$,
\begin{equation}
C_D(\tau,\varepsilon_{\rm in},\varepsilon_{\rm out};g,\phi) =
A(\mathcal V_{g,\phi},\psi_\tau) \le
A_D(\psi_\tau).
\label{eq:deterministic_capacity_upper_main}
\end{equation}
If, in addition, every direction in $\mathbb S^{D-1}$ is deterministically realizable, i.e., for every $u\in\mathbb S^{D-1}$ there exists $c(u)\in\mathcal Z$ such that $\phi(g(c(u),U))=u$ a.s., then
\begin{equation}
C_D(\tau,\varepsilon_{\rm in},\varepsilon_{\rm out};g,\phi) = A_D(\psi_\tau),
\quad \forall\, \varepsilon_{\rm in},\varepsilon_{\rm out}\in[0,1).
\end{equation}
\end{proposition}

\begin{proof}
A proof is provided in the Appendix.
\end{proof}

\subsection{Centered Identity Distributions}
\label{subsec:centered_identity_distributions}

Under nuisance variation, an identity-conditional embedding distribution need not be a point mass. For $u\in\mathbb S^{D-1}$ and $\rho\in[0,\pi]$, define the geodesic cap
\begin{equation}
\mathrm{Cap}(u,\rho) \triangleq \{x\in\mathbb S^{D-1}:\angle(x,u)\le \rho\}.
\label{eq:cap_definition_main}
\end{equation}

\begin{definition}[$(\rho,\eta)$-centered identity-conditional embedding distribution]
\label{def:centered_identity_distribution_main}
Fix $\rho\in[0,\pi]$ and $\eta\in[0,1)$. An identity-conditional embedding distribution $P_i$ is \textit{$(\rho,\eta)$-centered at} $u_i\in\mathbb S^{D-1}$ if
\begin{equation}
P_i\bigl(\mathrm{Cap}(u_i,\rho)\bigr)\ge 1-\eta.
\label{eq:centered_identity_condition_main}
\end{equation}
It is called \textit{$(\rho,\eta)$-centered} if such a $u_i$ exists.
\end{definition}

The parameter $\rho$ is the cap radius, and $\eta$ upper-bounds the probability mass outside the cap. Point-mass distributions are included. A point mass at $u$ is $(0,0)$-centered at $u$. Centeredness does not depend on the verification threshold. The strict inequality $2\rho<\arccos(\tau)$, imposed below, ensures that any two points in the same radius-$\rho$ cap have cosine score strictly greater than $\tau$.

\begin{definition}[$(\rho,\eta)$-restricted identity capacity]
\label{def:restricted_capacity_main}
For fixed $\rho\in[0,\pi)$ and $\eta\in[0,1)$, define
\begin{multline}
C_D^{(\rho,\eta)} (\tau,\varepsilon_{\rm in},\varepsilon_{\rm out};g,\phi)\\
\triangleq \sup\Bigl(\{0\}\cup\Bigl\{ M\in\mathbb N: \exists\,c_1,\ldots,c_M\in\mathcal Z
\text{ pairwise distinct},\\
\text{s.t.}\quad
\{P_{c_i}\}_{i=1}^M
\text{ is }
(\tau,\varepsilon_{\rm in},\varepsilon_{\rm out})\text{-admissible}\\
\text{and each }P_{c_i}\text{ is }(\rho,\eta)\text{-centered}
\Bigr\}\Bigr).
\label{eq:restricted_capacity_definition_main}
\end{multline}
\end{definition}

Definition~\ref{def:restricted_capacity_main} restricts the fixed-code operational capacity to families satisfying the specified centeredness condition. Consequently,
\begin{equation}
C_D^{(\rho,\eta)} (\tau,\varepsilon_{\rm in},\varepsilon_{\rm out};g,\phi)
\le C_D(\tau,\varepsilon_{\rm in},\varepsilon_{\rm out};g,\phi).
\label{eq:restricted_le_general_main}
\end{equation}

\begin{theorem}[Sufficient condition for admissibility]
\label{thm:sufficient_admissibility_main}
Suppose that each identity-conditional embedding distribution $P_i$ is $(\rho,\eta)$-centered at $u_i\in\mathbb S^{D-1}$ and that
\begin{align}
2\rho &<\arccos(\tau),
\label{eq:intra_condition_geometry_main}
\\
\angle(u_i,u_j) &\ge \arccos(\tau)+2\rho, \qquad i\ne j.
\label{eq:inter_condition_geometry_main}
\end{align}
Then the family $\{P_i\}_{i=1}^M$ is
\[
\bigl( \tau,\, 1-(1-\eta)^2,\, 1-(1-\eta)^2
\bigr)\text{-admissible}.
\]
\end{theorem}

\begin{proof}
A proof is provided in the Appendix.
\end{proof}

Let $\varepsilon_\eta\triangleq1-(1-\eta)^2$. Then
\[
2\varepsilon_\eta<1 \quad\Longleftrightarrow\quad \eta<1-\frac{1}{\sqrt{2}},
\]
so the tolerance in Theorem~\ref{thm:sufficient_admissibility_main} is in the nondegenerate regime exactly under this condition.

\begin{remark}
Define
\begin{equation}
\psi_\tau(\rho) \triangleq \arccos(\tau)+2\rho.
\label{eq:effective_separation_main}
\end{equation}
Then \eqref{eq:inter_condition_geometry_main} is equivalently
\[
\angle(u_i,u_j)\ge\psi_\tau(\rho), \qquad i\ne j.
\]
\end{remark}

\subsection{Cap-Volume Upper Bound for Spherical Codes}

Let $\sigma_{D-1}$ denote the standard rotation-invariant surface measure on $\mathbb S^{D-1}$. For $\alpha\in[0,\pi]$ and any $u\in\mathbb S^{D-1}$, define
\begin{equation}
V_D(\alpha) \triangleq \frac{\sigma_{D-1}(\mathrm{Cap}(u,\alpha))}{\sigma_{D-1}(\mathbb{S}^{D-1})}.
\label{eq:cap_measure_main}
\end{equation}
By rotational invariance, $V_D(\alpha)$ does not depend on the choice of $u$. Equivalently,
\begin{equation}
V_D(\alpha) = \frac{\int_0^\alpha \sin^{D-2}\theta\,d\theta}{\int_0^\pi \sin^{D-2}\theta\,d\theta}.
\label{eq:cap_measure_integral_main}
\end{equation}

\begin{proposition}[Cap-volume upper bound for spherical codes]
\label{prop:cap_volume_converse_main}
For every $\psi\in(0,\pi]$,
\begin{equation}
A_D(\psi)\le \frac{1}{V_D(\psi/2)}.
\label{eq:cap_volume_upper_main}
\end{equation}
\end{proposition}

\begin{proof}
A proof is provided in the Appendix.
\end{proof}

\subsection{Achievable Lower Bounds under Angular Expressivity}
\label{subsec:angular_expressivity}

To derive achievability bounds, we require that the generator--recognizer pair can realize centered identity-conditional distributions at prescribed directions on the sphere.

\begin{definition}[Full $(\rho,\eta)$-angular expressivity]
\label{def:full_angular_expressivity_main}
The pair $(g,\phi)$ is \textit{fully $(\rho,\eta)$-angularly expressive} if, for every finite set of pairwise-distinct directions $u_1,\ldots,u_M\in\mathbb S^{D-1}$, there exist pairwise-distinct latent codes $c_1,\ldots,c_M\in\mathcal Z$ such that $P_{c_i}$ is $(\rho,\eta)$-centered at $u_i$ for every $i\in[M]$.
\end{definition}

\begin{remark}[Interpretation of angular expressivity]
Full $(\rho,\eta)$-angular expressivity is a global realizability assumption used to derive spherical-code achievability bounds. It is not assumed for the pipelines used in the finite-sample evaluation.
\end{remark}

\begin{theorem}[Achievability under full angular expressivity]
\label{thm:achievability_main}
Assume that $(g,\phi)$ is fully $(\rho,\eta)$-angularly expressive and that $2\rho<\arccos(\tau)$. Then
\begin{equation}
C_D^{(\rho,\eta)} \Bigl( \tau,\, 1-(1-\eta)^2,\, 1-(1-\eta)^2; g,\phi \Bigr)
\ge A_D\bigl(\psi_\tau(\rho)\bigr),
\label{eq:achievability_main}
\end{equation}
where $\psi_\tau(\rho)$ is defined in \eqref{eq:effective_separation_main}.
\end{theorem}

\begin{proof}
A proof is provided in the Appendix.
\end{proof}

\begin{remark}[Heterogeneous identity radii]
Suppose that identity $i$ is $(\rho_i,\eta)$-centered at $u_i$, with a common $\eta\in[0,1)$. The argument of Theorem~\ref{thm:sufficient_admissibility_main} shows that the following conditions are sufficient:
\begin{align}
2\rho_i &<\arccos(\tau), \qquad &&\forall i,\\
\angle(u_i,u_j) &\ge\arccos(\tau)+\rho_i+\rho_j,
\qquad &&\forall i\ne j.
\end{align}
Under these conditions, the family is
\[
\bigl( \tau,\, 1-(1-\eta)^2,\, 1-(1-\eta)^2
\bigr)\text{-admissible}.
\]
This heterogeneous form is used in the finite-sample geometry check below.
\end{remark}

Using the finite sample $\mathcal S$ from Section~\ref{subsec:finite_sample_evaluation}, define the finite-sample mean embedding $\widehat m_i \triangleq \frac{1}{K_i}\sum_{k=1}^{K_i}E_i^{(k)}$. Fix $\zeta=10^{-12}$. For identities satisfying $\|\widehat m_i\|_2>\zeta$, define the empirical mean direction
\begin{equation}
\widehat u_i \triangleq \frac{\widehat m_i}{\|\widehat m_i\|_2}.
\label{eq:empirical_mean_direction_main}
\end{equation}
If $\|\widehat m_i\|_2\le\zeta$, $\widehat u_i$ and the corresponding empirical radius are undefined. For every identity for which $\widehat u_i$ is defined, let $a_{i,k} \triangleq \angle(E_i^{(k)},\widehat u_i)$, $k\in[K_i]$, and let $a_{i,(1)}\le\cdots\le a_{i,(K_i)}$ denote the corresponding order statistics. Define
\begin{equation}
\widehat\rho_i \triangleq a_{i,(\lceil(1-\eta)K_i\rceil)},
\label{eq:empirical_radius_main}
\end{equation}
where $\eta\in[0,1)$ is the empirical coverage-tail parameter. Then $\widehat\rho_i$ is the smallest observed angular radius about $\widehat u_i$ containing at least a fraction $1-\eta$ of the evaluated views. The operational rates in \eqref{eq:empirical_genuine_rate}--\eqref{eq:empirical_impostor_rate} remain well defined when $\widehat u_i$ is undefined.

The corresponding finite-sample geometric check requires
\begin{equation}
\widehat\rho_i+\widehat\rho_j+\arccos(\tau) \le \angle(\widehat u_i,\widehat u_j)
\label{eq:heterogeneous_empirical_sufficient}
\end{equation}
for every distinct pair of identities included in the check, together with $2\widehat\rho_i<\arccos(\tau)$ for every included identity. These conditions are used only for a finite-sample geometric check. Without a statistical coverage argument, they do not imply population-level admissibility. They also do not, by themselves, imply the empirical admissibility constraints defining $\widehat C_{\mathcal S}$.

\subsection{Asymptotic Lower Bounds}
\label{subsec:asymptotic_growth_main}

We next derive finite-dimensional and asymptotic lower bounds for $A_D(\psi)$ and the corresponding identity capacities.

\begin{proposition}[Covering lower bound]
\label{prop:covering_lower_bound_main}
For every $\psi\in(0,\pi)$,
\begin{equation}
A_D(\psi)\ge \frac{1}{V_D(\psi)}.
\label{eq:covering_lower_bound_main}
\end{equation}
\end{proposition}

\begin{proof}
A proof is provided in the Appendix.
\end{proof}

For every $\psi\in(0,\pi)$, Propositions~\ref{prop:cap_volume_converse_main} and~\ref{prop:covering_lower_bound_main} give
\begin{equation}
\frac{1}{V_D(\psi)} \le A_D(\psi)
\le \frac{1}{V_D(\psi/2)}.
\label{eq:packing_sandwich_main}
\end{equation}
Using the fixed-angle asymptotic for spherical-cap measure gives the following lower bounds; see also~\cite{kabatiansky1978bounds}.

\begin{theorem}[Asymptotic spherical-code and capacity lower bounds]
\label{thm:asymptotic_lower_bound_main}
For every $\psi\in(0,\pi/2)$,
\begin{equation}
\liminf_{D\to\infty} \frac{1}{D}\log A_D(\psi) \ge -\log(\sin\psi).
\label{eq:spherical_code_rate_lower_main}
\end{equation}
Consequently, let $\mathbf P=\{(g_D,\phi_D):D\in\mathcal D\}$ be a dimension-indexed family, and fix $\rho\in[0,\pi)$, $\eta\in[0,1)$, and $\tau\in[0,1)$. Suppose that, for every sufficiently large $D\in\mathcal D$, $(g_D,\phi_D)$ is fully $(\rho,\eta)$-angularly expressive, $2\rho<\arccos(\tau)$, and
$\psi_\tau(\rho)\in(0,\pi/2)$. Then
\begin{multline}
\liminf_{\substack{D\to\infty\\D\in\mathcal D}} \frac{1}{D}\log
C_D^{(\rho,\eta)}\Bigl(\tau,1-(1-\eta)^2,  1-(1-\eta)^2;g_D,\phi_D\Bigr) \\
 \ge -\log \bigl(\sin\psi_\tau(\rho)\bigr).
\label{eq:restricted_capacity_rate_lower_main}
\end{multline}
\end{theorem}

\begin{proof}
A proof is provided in the Appendix.
\end{proof}

\begin{corollary}[Deterministic view-invariant asymptotic lower bound]
\label{cor:deterministic_asymptotic_main}
Let $\mathbf P=\{(g_D,\phi_D):D\in\mathcal D\}$ be a dimension-indexed family. Suppose that, for every sufficiently large $D\in\mathcal D$, the embedding output is invariant to nuisance variation for every latent code and $\mathcal V_{g_D,\phi_D}=\mathbb S^{D-1}$. Then, for every $\tau\in(0,1)$ and $\varepsilon_{\rm in},\varepsilon_{\rm out}\in[0,1)$,
\begin{align}
R(\tau,\varepsilon_{\rm in},\varepsilon_{\rm out};\mathbf P)
&= \limsup_{\substack{D\to\infty\\D\in\mathcal D}} \frac{1}{D}
\log A_D\bigl(\arccos(\tau)\bigr) \nonumber\\
&\ge -\log \bigl(\sin(\arccos(\tau))\bigr) \nonumber\\
&= -\frac{1}{2}\log(1-\tau^2) >0.
\label{eq:deterministic_asymptotic_rate_main}
\end{align}
\end{corollary}

\begin{proof}
A proof is provided in the Appendix.
\end{proof}

Theorem~\ref{thm:asymptotic_lower_bound_main} yields a positive asymptotic lower-bound exponent for dimension-indexed families satisfying its uniform assumptions and $\psi_\tau(\rho)<\pi/2$. Within this regime, increasing $\rho$ or decreasing $\tau$ increases the sufficient center-separation threshold $\psi_\tau(\rho)$ and decreases the lower-bound exponent $-\log(\sin\psi_\tau(\rho))$.

\section{Prior-Constrained Random-Code Identity Capacity}
\label{sec:random_code_capacity}

By Proposition~\ref{prop:random_vs_deterministic_capacity}, the prior-constrained random-code capacity is upper-bounded by the corresponding fixed-code capacity.

\subsection{Random-Code Admissibility via Center Separation}

The random-code analysis below uses the centered model of Definition~\ref{def:centered_identity_distribution_main}.

\begin{assumption}[Common centered-distribution parameters under random coding]
\label{ass:random_centered_model}
There exist $\rho\in[0,\pi)$, $\eta\in[0,1)$, and a measurable map $u:\mathcal Z\to\mathbb S^{D-1}$ such that, for every $c\in\mathcal Z$, the induced identity-conditional embedding distribution $P_c$ is $(\rho,\eta)$-centered at $u(c)$.
\end{assumption}

Under Assumption~\ref{ass:random_centered_model}, a random latent code $C\sim P_C$ induces a random identity center $U\triangleq u(C)\in\mathbb S^{D-1}$. Let $Q$ denote the induced distribution of $U$ on $\mathbb S^{D-1}$; equivalently, $Q(B)=P_C\bigl(u^{-1}(B)\bigr)$ for every Borel set $B\subseteq\mathbb S^{D-1}$.

\begin{proposition}[Random-code admissibility via center separation]
\label{prop:random_centered_admissibility_main}
Suppose Assumption~\ref{ass:random_centered_model} holds and
\begin{equation}
2\rho<\arccos(\tau).
\label{eq:random_intra_condition}
\end{equation}
Let $C_1,\ldots,C_M\stackrel{\mathrm{i.i.d.}}{\sim}P_C$ and define
\begin{equation}
U_i\triangleq u(C_i),\qquad i\in[M].
\end{equation}
If
\begin{equation}
\angle(U_i,U_j)\ge\psi_\tau(\rho), \qquad \forall i\ne j,
\label{eq:random_pairwise_center_condition}
\end{equation}
where $\psi_\tau(\rho)$ is defined in \eqref{eq:effective_separation_main}, then the family $\{P_{C_i}\}_{i=1}^M$ is
\[
\bigl( \tau,\, 1-(1-\eta)^2,\, 1-(1-\eta)^2 \bigr)\text{-admissible}.
\]
Consequently,
\begin{align}
&\Pr \left[ \mathcal D_M\cap \mathcal A_M \left(
1-(1-\eta)^2,\, 1-(1-\eta)^2 \right) \right]
\nonumber\\
&\qquad \ge
\Pr \left[ \angle(U_i,U_j)\ge\psi_\tau(\rho), \, \forall i\ne j
\right].
\label{eq:random_admissibility_via_centers_main}
\end{align}
\end{proposition}

\begin{proof}
A proof is provided in the Appendix.
\end{proof}

Proposition~\ref{prop:random_centered_admissibility_main} provides a sufficient condition for random-code admissibility in terms of pairwise separation of the induced identity centers.

\begin{definition}[Pairwise separation-failure probability]
\label{def:center_collision_probability}
For $\psi>0$, define
\begin{equation}
q_Q(\psi) \triangleq \Pr \left[\angle(U_1,U_2)<\psi\right],
\label{eq:center_collision_probability}
\end{equation}
where $U_1,U_2\stackrel{\mathrm{i.i.d.}}{\sim}Q$.
\end{definition}

Thus, $q_Q(\psi)$ is the probability that two independent identity centers drawn from $Q$ have angular separation strictly smaller than $\psi$.

\subsection{High-Probability Lower Bounds}

\begin{theorem}[Union-bound lower bound for random-code capacity]
\label{thm:random_capacity_union_bound_main}
Suppose Assumption~\ref{ass:random_centered_model} holds and $2\rho<\arccos(\tau)$. Let $Q$ be the induced distribution of identity centers on $\mathbb S^{D-1}$. Then, for every integer $M\ge2$,
\begin{multline}
\Pr\bigl[ \mathcal D_M\cap \mathcal A_M\bigl(1-(1-\eta)^2,1-(1-\eta)^2\bigr) \bigr] \\
\ge 1-\binom{M}{2}\, q_Q\bigl(\psi_\tau(\rho)\bigr).
\label{eq:union_bound_admissibility_main}
\end{multline}
Consequently, for any $\delta\in(0,1)$, if
\begin{equation}
\binom{M}{2}\, q_Q\bigl(\psi_\tau(\rho)\bigr) \le\delta,
\label{eq:union_bound_capacity_condition_main}
\end{equation}
then
\begin{equation}
C_{D,\delta}^{\mathrm{rnd}} \Bigl(
\tau,\, 1-(1-\eta)^2,\, 1-(1-\eta)^2;\, g,\phi,P_C \Bigr) \ge M.
\label{eq:random_capacity_lower_bound_main}
\end{equation}
\end{theorem}

\begin{proof}
A proof is provided in the Appendix.
\end{proof}

Theorem~\ref{thm:random_capacity_union_bound_main} converts the pairwise separation-failure probability into a sufficient condition for a random codebook of size $M$ to be admissible with probability at least $1-\delta$.

\subsection{Asymptotic Lower Bounds}

We first consider the case in which the induced center distribution is uniform on the sphere.

\begin{corollary}[Uniform induced center distribution]
\label{cor:uniform_center_distribution_main}
Suppose Assumption~\ref{ass:random_centered_model} holds, $2\rho<\arccos(\tau)$, and the induced center distribution $Q$ is uniform on $\mathbb S^{D-1}$. If $\psi_\tau(\rho)\in(0,\pi]$, then
\begin{equation}
q_Q\bigl(\psi_\tau(\rho)\bigr) = V_D\bigl(\psi_\tau(\rho)\bigr),
\label{eq:uniform_collision_probability_main}
\end{equation}
where $V_D$ is defined in \eqref{eq:cap_measure_main}. Consequently, for every integer $M\ge2$ and every $\delta\in(0,1)$, if
\begin{equation}
\binom{M}{2} V_D\bigl(\psi_\tau(\rho)\bigr) \le\delta,
\label{eq:uniform_capacity_condition_main}
\end{equation}
then
\begin{equation}
C_{D,\delta}^{\mathrm{rnd}} \Bigl( \tau,\, 1-(1-\eta)^2,\,
1-(1-\eta)^2;\, g,\phi,P_C \Bigr) \ge M.
\label{eq:uniform_random_capacity_lower_bound_main}
\end{equation}
\end{corollary}

\begin{proof}
A proof is provided in the Appendix.
\end{proof}

\begin{theorem}[Asymptotic random-code lower bound under uniform induced centers]
\label{thm:asymptotic_random_capacity_main}
Fix $\rho\in[0,\pi)$, $\eta\in[0,1)$, $\tau\in[0,1)$, and $\delta\in(0,1)$. Let $\{(g_D,\phi_D,P_{C,D}):D\in\mathcal D\}$ be a dimension-indexed family. Suppose that, for every sufficiently large $D\in\mathcal D$, Assumption~\ref{ass:random_centered_model} holds with parameters $(\rho,\eta)$, the induced center distribution $Q_D$ is uniform on $\mathbb S^{D-1}$, $2\rho<\arccos(\tau)$, and  $\psi_\tau(\rho)\in(0,\pi/2)$. Then
\begin{multline}
\liminf_{\substack{D\to\infty\\D\in\mathcal D}} \frac{1}{D}
\log C_{D,\delta}^{\mathrm{rnd}} \Bigl( \tau,1-(1-\eta)^2,1-(1-\eta)^2;\\ 
g_D,\phi_D,P_{C,D} \Bigr) \ge -\frac{1}{2} \log\bigl(\sin\psi_\tau(\rho)\bigr).
\label{eq:asymptotic_random_capacity_main}
\end{multline}
\end{theorem}

\begin{proof}
A proof is provided in the Appendix.
\end{proof}

\section{Image-Level Evaluation}
\label{sec:image_level_evaluation}

We evaluate the finite-sample cardinality $\widehat C_{\mathcal S}$ defined in Section~\ref{subsec:finite_sample_evaluation} on released repeated-view DigiFace-1M \cite{bae2023digiface1m} images using three fixed recognizers. For each recognizer, the primary evaluation computes a maximum-cardinality subset of the fixed evaluation identity set satisfying all empirical genuine acceptance and pairwise impostor non-match requirements at the recognizer-specific verification threshold obtained from the specified calibration protocol. The DigiFace generator is neither retrained nor rerun, and no new images or identities are generated. Without additional sampling and latent-space coverage assumptions, this sample-restricted cardinality is neither an estimator nor a confidence bound for the fixed-code capacity $C_D$.

\subsection{Protocol and In-Domain Calibration}

DigiFace-1M is used because the released data provide 72 rendered views per generator-defined identity, permitting repeated-view evaluation of the empirical genuine and pairwise impostor quantities. Using released images fixes the candidate identity and image pool without rerunning the generator. From the P1 block~\cite{bae2023digiface1m}, a deterministic selection with seed 20260901 chooses 250 identities and 24 images per identity, yielding 6,000 selected source images. We do not claim that this selected subset is representative of DigiFace-1M.

Before recognizer scoring, we evaluate YuNet~\cite{wu2023yunet} detector score thresholds $0.90$, $0.85$, $0.80$, $0.75$, and $0.70$. We use $0.80$, the largest evaluated threshold for which every selected identity retains at least 12 of its 24 views. The common YuNet detection and five-landmark $112\times112$ alignment pipeline~\cite{wu2023yunet,deng2019arcface} retains 5,868 images from all 250 identities, with 13--24 retained views per identity. The same 5,868 aligned crops are supplied to OpenCV SFace~\cite{zhong2021sface}, a TFace IR-50 checkpoint, and the ADMIS Syn-30k IR-50/ArcFace checkpoint~\cite{deng2019arcface,deandres2024frcsyn}.

For the primary evaluation, the 250 identities are split once at the identity level, with 50 identities used to calibrate the verification threshold and the remaining 200 identities used only for evaluation. The same seed-20260901 split is applied to all recognizers. The calibration and evaluation sets contain disjoint identities and source images. A separate verification threshold is calibrated for each recognizer to a target empirical calibration FMR of $10^{-3}$ on the 50 calibration identities. We set $\varepsilon_{\rm in}=\varepsilon_{\rm out}=0.05$ for the operational graph, while $\eta=0.05$ is used only for the separate empirical geometric check based on empirical mean directions and radii. The TFace/ADMIS comparison is descriptive across two fixed recognition pipelines and does not isolate the effect of training-data source, because the public checkpoints are not controlled training runs differing only in training data.

\begin{figure*}[t]
\centering
\includegraphics[width=0.85\textwidth]{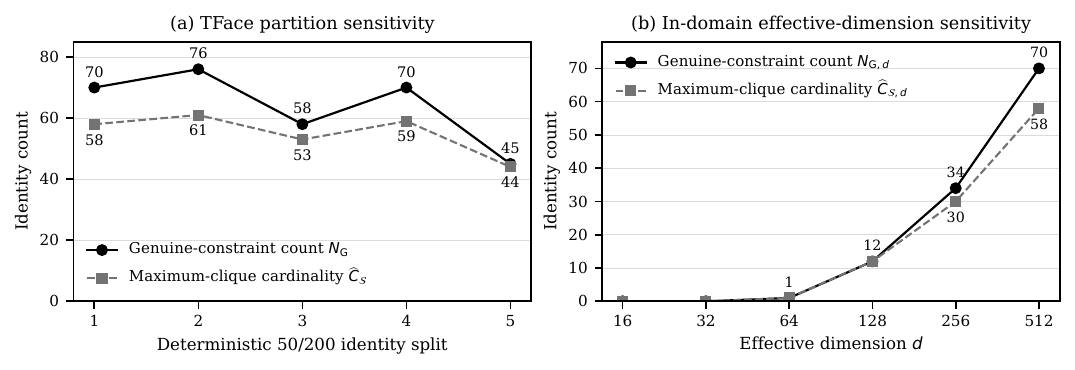}
\vspace{-4pt}
\caption{DigiFace in-domain sensitivity analyses.}
\label{fig:image_level_evaluation}
\vspace{-6pt}
\end{figure*}

\subsection{Cross-Recognizer Results at a Common Calibration Target}

\begin{table}[!b]
\centering
\caption{DigiFace in-domain cross-recognizer evaluation.}
\label{tab:real_pipeline_summary}
\footnotesize
\setlength{\tabcolsep}{3.8pt}
\begin{tabular}{@{}lcccccc@{}}
\toprule
Recognizer
& $D$
& $\tau$
& $\alpha$
& $N_{\rm eval}$
& $N_{\rm G}$
& $\widehat C_{\mathcal S}$ \\
\midrule
ADMIS Syn-30k IR-50
& 512 & 0.4964 & $10^{-3}$ & 200 & 9 & 9 \\
SFace
& 128 & 0.5227 & $10^{-3}$ & 200 & 2 & 2 \\
TFace IR-50
& 512 & 0.3985 & $10^{-3}$ & 200 & 70 & 58 \\
\bottomrule
\end{tabular}
\end{table}

Table~\ref{tab:real_pipeline_summary} lists the primary in-domain cross-recognizer results. The parameter $\alpha$ is the target empirical calibration FMR, $N_{\rm eval}$ is the number of evaluation identities, $N_{\rm G}$ is the number satisfying the empirical genuine constraint, and $\widehat C_{\mathcal S}$ is the exact maximum-clique cardinality of the corresponding compatibility graph. The threshold $\tau$ is calibrated separately for each recognizer.

At the primary in-domain calibration target, TFace has $N_{\rm G}=70$ and exact maximum-clique cardinality 58. Thus, a maximum compatible subset contains $58/70\approx82.9\%$ of the TFace identities satisfying the empirical genuine constraint. For SFace and ADMIS Syn-30k, $N_{\rm G}$ equals 2 and 9, respectively, and in each case these identities form a complete compatibility graph. Consequently, for these two recognizers, $\widehat C_{\mathcal S}=N_{\rm G}$ on the evaluated split. With $\eta=0.05$, no unordered pair of the 200 evaluation identities satisfies the finite-sample geometric check for any of the three recognizers. Failure of this finite-sample geometric check does not imply failure of the empirical admissibility constraints.

To assess sensitivity to the identity partition, we repeat the 50/200 protocol with seeds 20260901--20260905. The corresponding $N_{\rm G}$ values are 70, 76, 58, 70, and 45, and the exact maximum-clique cardinalities are 58, 61, 53, 59, and 44. The ratios $\widehat C_{\mathcal S}/N_{\rm G}$ therefore range from 80.3\% to 97.8\% across the five evaluated splits. Figure~\ref{fig:image_level_evaluation}(a) shows these split-specific values.

\subsection{Operating-Point and Tolerance Sensitivity}

On the primary TFace split, target empirical calibration FMRs $10^{-3}$, $2\times10^{-3}$, $5\times10^{-3}$, and $10^{-2}$ give $N_{\rm G}$ values of 70, 83, 103, and 135, respectively. The exact maximum-clique cardinalities are 58 and 62 at the first two operating points. At the latter two operating points, the certified lower bounds on the maximum-clique cardinality are 54 and 55, respectively. Because the latter two values are lower bounds, these results do not establish monotonicity of the exact sample-restricted cardinality with respect to the calibration target.

At the primary TFace threshold, fixing $\varepsilon_{\rm in}=0.05$ and setting $\varepsilon_{\rm out}$ to $0.01$, $0.05$, and $0.10$ gives exact maximum-clique cardinalities of 43, 58, and 64, respectively; $N_{\rm G}=70$ at all three values. Fixing $\varepsilon_{\rm out}=0.05$ and setting $\varepsilon_{\rm in}$ to $0.02$, $0.05$, and $0.10$ gives $(N_{\rm G},\widehat C_{\mathcal S})=(34,31)$ and $(70,58)$ at the first two values, respectively, while the last value gives $N_{\rm G}=103$ and $\widehat C_{\mathcal S}\ge82$. The sample-restricted cardinality therefore varies across both tolerance sweeps. These analyses do not specify application-independent choices of $\alpha$, $\varepsilon_{\rm in}$, or $\varepsilon_{\rm out}$.

A separate threshold-transfer check calibrates each recognizer on LFW~\cite{huang2008labeled} at target empirical FMR $10^{-3}$ and applies the resulting threshold to DigiFace without recalibration. The resulting aggregate complete-pair empirical DigiFace impostor FMRs are approximately $7.33\%$, $2.55\%$, and $9.06\%$ for SFace, ADMIS Syn-30k IR-50, and TFace IR-50, respectively. The LFW-calibrated thresholds therefore yield empirical DigiFace FMRs above the target value of $10^{-3}$ on the selected 250-identity sample and are not interpreted as in-domain DigiFace operating points at empirical FMR $10^{-3}$.

\subsection{Effective-Dimension Sensitivity}

For the 512-dimensional TFace representation, a single orthonormal basis generated with seed 20260901 defines nested subspaces of dimensions $d\in\{16,32,64,128,256,512\}$, where the $d$-dimensional subspace is spanned by the first $d$ basis vectors. Embeddings from the same 50 calibration identities and 200 evaluation identities are projected onto each subspace, renormalized to unit Euclidean norm, and recalibrated separately to a target empirical calibration FMR of $10^{-3}$ at each $d$. Table~\ref{tab:effective_dimension_sensitivity} lists the resulting finite-sample quantities.

\begin{table}[!b]
\centering
\caption{TFace effective-dimension sensitivity.}
\label{tab:effective_dimension_sensitivity}
\footnotesize
\begin{tabular}{ccccc}
\toprule
$d$
& $\tau_d$
& $\alpha$
& $N_{{\rm G},d}$
& $\widehat C_{\mathcal S,d}$ \\
\midrule
16  & 0.7776 & $10^{-3}$ & 0  & 0  \\
32  & 0.6417 & $10^{-3}$ & 0  & 0  \\
64  & 0.5445 & $10^{-3}$ & 1  & 1  \\
128 & 0.4838 & $10^{-3}$ & 12 & 12 \\
256 & 0.4338 & $10^{-3}$ & 34 & 30 \\
512 & 0.3985 & $10^{-3}$ & 70 & 58 \\
\bottomrule
\end{tabular}
\end{table}

The parameter $\tau_d$ is the verification threshold calibrated after projection to dimension $d$, $N_{{\rm G},d}$ is the number of evaluation identities satisfying the empirical genuine constraint, and $\widehat C_{\mathcal S,d}$ is the exact maximum-clique cardinality of the corresponding compatibility graph. The target empirical calibration FMR is $\alpha=10^{-3}$ at every $d$.

The exact maximum-clique cardinalities are 0, 0, 1, 12, 30, and 58 for $d=16,32,64,128,256,512$, respectively, and the corresponding $N_{{\rm G},d}$ values are 0, 0, 1, 12, 34, and 70. Along the decreasing dimension sequence $d=512,256,128,64,32,16$, both quantities are nonincreasing. The pairs $(N_{{\rm G},d},\widehat C_{\mathcal S,d})$ are $(70,58)$, $(34,30)$, $(12,12)$, $(1,1)$, $(0,0)$, and $(0,0)$, respectively. All maximum-clique cardinalities in this analysis are exact. Figure~\ref{fig:image_level_evaluation}(b) shows the same dimension-dependent quantities.

This projection experiment measures sensitivity of one fixed 512-dimensional recognizer to dimensional compression followed by renormalization and threshold recalibration. It does not realize a family of independently specified pipelines with increasing native embedding dimension and therefore does not validate the dimension-indexed asymptotic results.

\section{Conclusion}
\label{sec:conclusion_main}

The number of nominal synthetic identities does not determine how many are jointly distinguishable under a specified verification rule. For deterministic view-invariant pipelines, fixed-code capacity is independent
of $\varepsilon_{\rm in},\varepsilon_{\rm out}\in[0,1)$ and equals the spherical-code cardinality over the realizable embedding set. Under centered stochastic variation, the analysis gives a sufficient
center-separation condition. Under full $(\rho,\eta)$-angular expressivity, it yields finite-dimensional and asymptotic capacity lower bounds rather than an exact characterization in general. The finite-sample maximum-clique cardinality is a sample-restricted joint-admissibility quantity and is not a population-capacity estimate without additional sampling and latent-space coverage assumptions. The theoretical latent-code supremum and the sample-restricted cardinality are therefore distinct quantities.

\vspace{-5pt}

\section*{Acknowledgment}

This work was supported by the Swiss National Science Foundation (SNSF) under Grant No.~222339.

 

\vspace{-5pt}

\putbib
\end{bibunit}

\clearpage
\appendices


\begin{bibunit}

\renewcommand{\thesubsection}{\thesection.\arabic{subsection}}
\renewcommand{\thesubsectiondis}{\thesection.\arabic{subsection}.}

\section{Detailed Proofs of the Main Results}
\label{sec:omitted-proofs}

\subsection{Proof of Proposition~\ref{prop:random_vs_deterministic_capacity}}
\label{app:proof_prop_random_vs_deterministic_capacity}

\begin{proof}
Suppose that
\begin{equation}
\Pr\bigl[\mathcal D_M\cap\mathcal A_M(\varepsilon_{\rm in},\varepsilon_{\rm out})\bigr] \ge 1-\delta.
\end{equation}
Since $\delta<1$, this probability is strictly positive. Hence there exists at least one realization $(c_1,\ldots,c_M)\in\mathcal Z^M$ for which the codes are pairwise distinct and the induced family $\{P_{c_i}\}_{i=1}^M$ is $(\tau,\varepsilon_{\rm in},\varepsilon_{\rm out})$-admissible. By Definition~\ref{def:operational_capacity_main}, this implies
\begin{equation}
M\le C_D(\tau,\varepsilon_{\rm in},\varepsilon_{\rm out};g,\phi).
\end{equation}
Taking the supremum over all such $M$ proves \eqref{eq:random_vs_deterministic_capacity}.
\end{proof}

\subsection{Proof of Proposition~\ref{prop:deterministic_capacity_main}}
\label{app:proof_prop_deterministic_capacity}

\begin{proof}
Fix $\varepsilon_{\rm in},\varepsilon_{\rm out}\in[0,1)$. Consider a pairwise-distinct codebook $(c_1,\ldots,c_M)$ and set $u_i\triangleq v(c_i)$. Under the view-invariance assumption, $E_i^{(k)}=u_i$ a.s. for every $i$ and $k$. Since $\tau<1$,
\begin{equation}
\Pr \left[ \left\langle E_i^{(1)},E_i^{(2)}\right\rangle>\tau
\right] =1, \qquad i\in[M].
\end{equation}
Hence the genuine constraint holds for every $\varepsilon_{\rm in}\in[0,1)$. For $i\ne j$,
\begin{equation}
\Pr \left[ \left\langle E_i^{(1)},E_j^{(1)}\right\rangle\le\tau \right]
= \mathbf 1 \left\{\langle u_i,u_j\rangle\le\tau\right\}.
\end{equation}
Because $\varepsilon_{\rm out}<1$, one has $1-\varepsilon_{\rm out}>0$. Therefore the impostor constraint
\[
\Pr \left[ \left\langle E_i^{(1)},E_j^{(1)}\right\rangle\le\tau \right]
\ge 1-\varepsilon_{\rm out}
\]
holds if and only if $\langle u_i,u_j\rangle\le\tau$, or equivalently,
\[
\angle(u_i,u_j)\ge\arccos(\tau)=\psi_\tau.
\]
Thus every admissible codebook induces a subset of $\mathcal V_{g,\phi}$ with minimum pairwise angle at least $\psi_\tau$. Conversely, let $u_1,\ldots,u_M\in\mathcal V_{g,\phi}$ satisfy $\angle(u_i,u_j)\ge\psi_\tau$ for all $i\ne j$. For each $i$, choose $c_i\in\mathcal Z$ realizing $u_i$. Since $\psi_\tau>0$, the directions $u_1,\ldots,u_M$ are pairwise distinct; consequently, the realizing codes $c_1,\ldots,c_M$ are also pairwise distinct. The resulting codebook satisfies both admissibility constraints. Hence
\begin{equation}
C_D(\tau,\varepsilon_{\rm in},\varepsilon_{\rm out};g,\phi) = A(\mathcal V_{g,\phi},\psi_\tau).
\end{equation}
Since $\mathcal V_{g,\phi}\subseteq\mathbb S^{D-1}$,
\[
A(\mathcal V_{g,\phi},\psi_\tau) \le A_D(\psi_\tau).
\]
If every direction in $\mathbb S^{D-1}$ is deterministically realizable, then $\mathcal V_{g,\phi}=\mathbb S^{D-1}$, and the final equality follows.
\end{proof}

\subsection{Proof of Theorem~\ref{thm:sufficient_admissibility_main}}
\label{app:proof_thm_sufficient_admissibility}

\begin{proof}
Fix $i$ and let
\begin{equation}
A_i^{(k)}\triangleq \{E_i^{(k)}\in \mathrm{Cap}(u_i,\rho)\}.
\end{equation}
By Definition~\ref{def:centered_identity_distribution_main}, $\Pr(A_i^{(k)})\ge 1-\eta$. On the event $A_i^{(1)}\cap A_i^{(2)}$, the spherical triangle inequality yields
\begin{align}
\angle(E_i^{(1)},E_i^{(2)})
& \le \angle(E_i^{(1)},u_i)+\angle(u_i,E_i^{(2)})\\
& \le 2\rho < \arccos(\tau).    
\end{align}
Hence $\langle E_i^{(1)},E_i^{(2)}\rangle > \tau$ whenever $A_i^{(1)}\cap A_i^{(2)}$ occurs. Therefore
\begin{equation}
\Pr \left[\langle E_i^{(1)},E_i^{(2)}\rangle>\tau\right]
\ge \Pr(A_i^{(1)}\cap A_i^{(2)})
\ge (1-\eta)^2.
\end{equation}
Now fix $i\neq j$. On the event $A_i^{(1)}\cap A_j^{(1)}$, one has
\begin{align}
\angle(E_i^{(1)},E_j^{(1)}) & \ge 
\angle(u_i,u_j)-\angle(E_i^{(1)},u_i)-\angle(E_j^{(1)},u_j) \nonumber\\
& \ge  \angle(u_i,u_j)-2\rho \nonumber\\
& \ge  \arccos(\tau),    
\end{align}
where the last step uses \eqref{eq:inter_condition_geometry_main}. Hence $\langle E_i^{(1)},E_j^{(1)}\rangle \le \tau$ whenever $A_i^{(1)}\cap A_j^{(1)}$ occurs, and therefore
\begin{equation}
\Pr \left[\langle E_i^{(1)},E_j^{(1)}\rangle\le\tau\right] \ge \Pr(A_i^{(1)}\cap A_j^{(1)})
\ge (1-\eta)^2.
\end{equation}
This proves the claim.
\end{proof}

\subsection{Proof of Proposition~\ref{prop:cap_volume_converse_main}}
\label{app:proof_prop_cap_volume_converse}

\begin{proof}
Let $\{u_1,\ldots,u_M\}\subset\mathbb{S}^{D-1}$ be any spherical code with minimum pairwise angle at least $\psi$. Then the interiors of the caps $\mathrm{Cap}(u_i,\psi/2)$, $i\in[M]$, are pairwise disjoint. Indeed, if two cap interiors intersected, then there would exist a point whose angular distance to both centers is strictly smaller than $\psi/2$. By the triangle inequality for geodesic distance on the sphere, the corresponding centers would then be at angular distance strictly smaller than $\psi$, contradicting the minimum-angle condition. Any intersections of the closed caps can therefore occur only on their boundaries, and hence have surface measure zero.
Therefore,
\begin{equation}
\sum_{i=1}^M \sigma_{D-1}\bigl(\mathrm{Cap}(u_i,\psi/2)\bigr)
\le \sigma_{D-1}(\mathbb{S}^{D-1}).
\end{equation}
By rotational invariance, all caps of radius $\psi/2$ have the same surface measure, so
\begin{equation}
M\,\sigma_{D-1}\bigl(\mathrm{Cap}(u,\psi/2)\bigr)
\le \sigma_{D-1}(\mathbb{S}^{D-1}).
\end{equation}
Dividing by $\sigma_{D-1}(\mathbb{S}^{D-1})$ yields
\begin{equation}
M\le \frac{1}{V_D(\psi/2)}.
\end{equation}
Since this holds for every spherical code with minimum angle at least $\psi$, it also holds for $A_D(\psi)$.
\end{proof}

\subsection{Proof of Theorem~\ref{thm:achievability_main}}
\label{app:proof_thm_achievability}

\begin{proof}
Let $\{u_1,\ldots,u_M\}\subset\mathbb{S}^{D-1}$ be a spherical code of size $M= A_D\bigl(\psi_\tau(\rho)\bigr)$ with minimum pairwise angle at least $\psi_\tau(\rho)$. Because $\psi_\tau(\rho)>0$, these directions are pairwise distinct. By full $(\rho,\eta)$-angular expressivity, there exist pairwise-distinct latent codes $c_1,\ldots,c_M$ whose induced distributions $P_i$ are $(\rho,\eta)$-centered at the corresponding $u_i$. Theorem~\ref{thm:sufficient_admissibility_main} then implies that the family $\{P_i\}_{i=1}^M$ is
\[
\bigl(\tau,1-(1-\eta)^2,1-(1-\eta)^2\bigr)\text{-admissible}.
\]
Hence the code size $M$ is achievable, which proves \eqref{eq:achievability_main}.
\end{proof}

\subsection{Proof of Proposition~\ref{prop:covering_lower_bound_main}}
\label{app:proof_prop_covering_lower_bound}

\begin{proof}
Let $\{u_1,\ldots,u_M\}$ be a spherical code of maximum cardinality $M=A_D(\psi)$. Such a code is maximal under inclusion. Hence the caps $\mathrm{Cap}(u_i,\psi)$ must cover the sphere; otherwise one could add another point without violating the minimum-angle constraint. Therefore
\[
\sigma_{D-1}(\mathbb{S}^{D-1}) \le M\,\sigma_{D-1}(\mathrm{Cap}(u,\psi)),
\]
which implies \eqref{eq:covering_lower_bound_main}.
\end{proof}

\subsection{Proof of Theorem~\ref{thm:asymptotic_lower_bound_main}}
\label{app:proof_thm_asymptotic_lower_bound}

\begin{proof}
By Proposition~\ref{prop:covering_lower_bound_main}, $A_D(\psi)\ge \frac{1}{V_D(\psi)}$. For every fixed $\alpha\in(0,\pi/2)$, the normalized spherical-cap measure satisfies
\begin{equation}
\lim_{D\to\infty}\frac{1}{D}\log V_D(\alpha)=\log(\sin\alpha).
\end{equation}
Equivalently,
\begin{equation}
V_D(\alpha) = \exp \left( -D\,[-\log(\sin\alpha)] + o(D) \right).
\end{equation}
Substituting $\alpha=\psi$ yields \eqref{eq:spherical_code_rate_lower_main}. Applying Theorem~\ref{thm:achievability_main} separately to each sufficiently large member $(g_D,\phi_D)$ of the dimension-indexed family and taking the liminf gives \eqref{eq:restricted_capacity_rate_lower_main}.
\end{proof}

\subsection{Proof of Corollary~\ref{cor:deterministic_asymptotic_main}}
\label{app:proof_cor_deterministic_asymptotic}

\begin{proof}
For every sufficiently large $D\in\mathcal D$, Proposition~\ref{prop:deterministic_capacity_main} gives
\[
C_D(\tau,\varepsilon_{\rm in},\varepsilon_{\rm out};g_D,\phi_D) = A_D(\arccos(\tau)).
\]
The equality in \eqref{eq:deterministic_asymptotic_rate_main} therefore follows from the definition in \eqref{eq:rate_definition_main}. Since $\arccos(\tau)\in(0,\pi/2)$ for $\tau\in(0,1)$, \eqref{eq:spherical_code_rate_lower_main} gives
\[
\liminf_{D\to\infty} \frac{1}{D}\log A_D(\arccos(\tau))
\ge -\log \bigl(\sin(\arccos(\tau))\bigr).
\]
Finally, $\sin(\arccos(\tau))=\sqrt{1-\tau^2}$, which yields the stated bound.
\end{proof}

\subsection{Proof of Proposition~\ref{prop:random_centered_admissibility_main}}
\label{app:proof_prop_random_centered_admissibility}

\begin{proof}
By Assumption~\ref{ass:random_centered_model}, each random distribution $P_{C_i}$ is $(\rho,\eta)$-centered at $U_i$. Together with \eqref{eq:random_intra_condition} and \eqref{eq:random_pairwise_center_condition}, the hypotheses of Theorem~\ref{thm:sufficient_admissibility_main} are satisfied. Hence the family $\{P_{C_i}\}_{i=1}^M$ is
\[
\bigl(\tau,\;1-(1-\eta)^2,\;1-(1-\eta)^2\bigr)\text{-admissible}.
\]
This proves the first claim. Moreover, on the center-separation event in \eqref{eq:random_pairwise_center_condition}, one has $U_i\ne U_j$ for every $i\ne j$ because $\psi_\tau(\rho)>0$. Since $U_i=u(C_i)$ is a deterministic function of the latent code, $C_i=C_j$ would imply $U_i=U_j$. Hence the center-separation event is contained in $\mathcal D_M$ as well as in the admissibility event, which proves \eqref{eq:random_admissibility_via_centers_main}.
\end{proof}

\subsection{Proof of Theorem~\ref{thm:random_capacity_union_bound_main}}
\label{app:proof_thm_random_capacity_union_bound}

\begin{proof}
Let
\begin{equation}
\mathcal E_{ij}\triangleq \{\angle(U_i,U_j)<\psi_\tau(\rho)\}, \qquad 1\le i<j\le M.
\end{equation}
Then
\begin{equation}
\Bigl\{ \angle(U_i,U_j)\ge \psi_\tau(\rho),\ \forall i\neq j \Bigr\}
= \Bigl( \bigcup_{1\le i<j\le M}\mathcal E_{ij} \Bigr)^c.
\end{equation}
By the union bound,
\begin{equation}
\Pr \Bigl[ \bigcup_{1\le i<j\le M}\mathcal E_{ij} \Bigr]
\le \sum_{1\le i<j\le M}\Pr(\mathcal E_{ij})
= \binom{M}{2}\,q_Q \bigl(\psi_\tau(\rho)\bigr),
\end{equation}
because each pair $(U_i,U_j)$ has the same distribution as $(U_1,U_2)$. Therefore,
\begin{equation}
\Pr \Bigl[ \angle(U_i,U_j)\ge \psi_\tau(\rho),\ \forall i\neq j \Bigr]
\ge 1-\binom{M}{2}\,q_Q \bigl(\psi_\tau(\rho)\bigr).
\end{equation}
On the displayed separation event, $U_i\ne U_j$ for every $i\ne j$ because $\psi_\tau(\rho)>0$. Since equal latent codes induce the same center, the sampled codes are therefore pairwise distinct on this event. Combining this with Proposition~\ref{prop:random_centered_admissibility_main} yields \eqref{eq:union_bound_admissibility_main}. The lower bound \eqref{eq:random_capacity_lower_bound_main} then follows from Definition~\ref{def:random_capacity_main}.
\end{proof}

\subsection{Proof of Corollary~\ref{cor:uniform_center_distribution_main}}
\label{app:proof_cor_uniform_center_distribution}

\begin{proof}
Fix $U_1=u$. Under the uniform distribution on $\mathbb S^{D-1}$, the open and closed caps of angular radius $\psi_\tau(\rho)$ have the same normalized surface measure because their boundary has zero surface measure. Hence the probability that $U_2$ satisfies $\angle(u,U_2)<\psi_\tau(\rho)$ equals $V_D(\psi_\tau(\rho))$. By rotational invariance, this probability does not depend on $u$. Averaging over $U_1$ proves \eqref{eq:uniform_collision_probability_main}. The capacity bound then follows from Theorem~\ref{thm:random_capacity_union_bound_main}.
\end{proof}

\subsection{Proof of Theorem~\ref{thm:asymptotic_random_capacity_main}}
\label{app:proof_thm_asymptotic_random_capacity}

\begin{proof}
Let $\Lambda(\psi)\triangleq -\log(\sin\psi)$, $\psi\in(0,\pi/2)$. Fix $\epsilon\in\bigl(0,\Lambda(\psi_\tau(\rho))/2\bigr)$ and define
\begin{equation}
M_D \triangleq \left\lfloor \exp \left( D\Bigl(\frac{\Lambda(\psi_\tau(\rho))}{2}-\epsilon\Bigr)
\right) \right\rfloor.
\end{equation}
By Corollary~\ref{cor:uniform_center_distribution_main}, it suffices to verify that
\begin{equation}
\binom{M_D}{2}\,V_D \bigl(\psi_\tau(\rho)\bigr)\le \delta
\end{equation}
for all sufficiently large $D\in\mathcal D$. Using the fixed-angle asymptotic
\begin{equation}
V_D(\alpha) = \exp \bigl( -D\,\Lambda(\alpha)+o(D) \bigr), \qquad \alpha\in(0,\pi/2),
\end{equation}
we obtain
\begin{subequations}
\begin{align}
\log \binom{M_D}{2} & = 2D\Bigl(\frac{\Lambda(\psi_\tau(\rho))}{2}-\epsilon\Bigr)+o(D) \\
& = D\Lambda(\psi_\tau(\rho))-2\epsilon D+o(D).    
\end{align}    
\end{subequations}
Hence
\begin{equation}
\log \left( \binom{M_D}{2}\,V_D(\psi_\tau(\rho)) \right)
= -2\epsilon D+o(D),
\end{equation}
which tends to $-\infty$ as $D\to\infty$. Therefore,
\begin{equation}
\binom{M_D}{2}\,V_D(\psi_\tau(\rho))\to 0.
\end{equation}
In particular, \eqref{eq:uniform_capacity_condition_main} holds for all sufficiently large $D\in\mathcal D$. Thus
\begin{equation}
C_{D,\delta}^{\mathrm{rnd}} \Bigl( \tau,\,
1-(1-\eta)^2,\, 1-(1-\eta)^2;\, g_D,\phi_D,P_{C,D}
\Bigr) \ge M_D,
\end{equation}
for all sufficiently large $D\in\mathcal D$. Taking logarithms, dividing by $D$, and letting $\epsilon\downarrow 0$ yields \eqref{eq:asymptotic_random_capacity_main}.
\end{proof}

\vspace{10pt}

\section{Relation to Maximum Manifold Capacity Representations}
\label{sec:mmcr_surrogates}

Maximum Manifold Capacity Representations (MMCR) studies normalized multi-view embeddings through mean representations and a nuclear-norm objective~\cite{yerxa2023learning,isik2023information,
schaeffer2024towards}. For comparison, define the identity-conditional mean embedding $\mu_i\triangleq\mathbb E[E_i^{(1)}]\in\mathbb R^D$, $i\in[M]$, and the identity-mean matrix
\begin{equation}
\mathbf B\triangleq
\begin{bmatrix}
\mu_1^\top\\
\vdots\\
\mu_M^\top
\end{bmatrix}
\in\mathbb R^{M\times D}.
\label{eq:mean_embedding_matrix_main}
\end{equation}
For this comparison, the centeredness condition constrains within-identity embeddings through angular-cap concentration, whereas $\|\mathbf B\|_*$ is a spectral functional of the identity-mean matrix. Neither quantity is equivalent to the operational admissibility constraints in Definition~\ref{def:operational_capacity_main}.

\begin{proposition}[Nuclear-norm bound for the identity-mean matrix]
\label{prop:nuclear_norm_main}
Let $r=\operatorname{rank}(\mathbf B)$. Since $\|E_i^{(1)}\|_2=1$ almost surely, $\|\mu_i\|_2\le1$ for every $i\in[M]$. Then
\begin{equation}
\|\mathbf{B}\|_* \le \sqrt{r}\,\|\mathbf{B}\|_F
\le \sqrt{rM} \le
\begin{cases}
M, & M\le D,\\[3pt]
\sqrt{MD}, & M>D.
\end{cases}
\label{eq:nuclear_norm_bound_main}
\end{equation}
\end{proposition}

Proposition~\ref{prop:nuclear_norm_main} bounds the nuclear norm through both the matrix rank and the Frobenius norm, with $\|\mathbf B\|_F^2=\sum_i\|\mu_i\|_2^2$. This spectral quantity does not characterize the operational capacity defined by \eqref{eq:admissible_intra_main}--\eqref{eq:admissible_inter_main}.

\begin{proof}
Let $\sigma_1,\ldots,\sigma_r$ denote the nonzero singular values of $\mathbf{B}$. By Cauchy--Schwarz,
\begin{equation}
\|\mathbf{B}\|_* = \sum_{\ell=1}^r \sigma_\ell
\le \sqrt{r} \left( \sum_{\ell=1}^r \sigma_\ell^2 \right)^{1/2}
= \sqrt{r}\,\|\mathbf{B}\|_F.
\end{equation}
Moreover, $\|\mathbf{B}\|_F^2 = \sum_{i=1}^M \|\mu_i\|_2^2 \le M$. Hence $\|\mathbf{B}\|_* \le \sqrt{rM}$. Since $r\le \min\{M,D\}$, the stated bound follows.
\end{proof}

The centeredness condition also yields a lower bound on the norm of the corresponding identity mean embedding.

\begin{proposition}[Mean-norm lower bound under centeredness]
\label{prop:concentration_mean_norm}
Let $P_i$ be a $(\rho,\eta)$-centered identity-conditional embedding distribution centered at $u_i\in\mathbb{S}^{D-1}$, and let $\mu_i \triangleq \mathbb{E}[E_i^{(1)}]$. Then
\begin{equation}
\langle \mu_i,u_i\rangle \ge (1-\eta)\cos\rho - \eta.
\label{eq:mean_center_alignment}
\end{equation}
Consequently,
\begin{equation}
\|\mu_i\|_2 \ge \bigl[(1-\eta)\cos\rho - \eta\bigr]_+,
\label{eq:mean_norm_lower_bound}
\end{equation}
where $[x]_+\triangleq\max\{x,0\}$.
\end{proposition}

\begin{proof}
Let $A_i\triangleq \{E_i^{(1)}\in \mathrm{Cap}(u_i,\rho)\}$. Then $\Pr(A_i)\ge 1-\eta$. On $A_i$, $\langle E_i^{(1)},u_i\rangle \ge \cos\rho$. On $A_i^c$, $\langle E_i^{(1)},u_i\rangle \ge -1$. Therefore
\begin{equation}
\langle \mu_i,u_i\rangle = \mathbb{E}\bigl[\langle E_i^{(1)},u_i\rangle\bigr] \ge (1-\eta)\cos\rho-\eta.
\end{equation}
Since $\|\mu_i\|_2\ge0$ and $\|\mu_i\|_2\ge \langle \mu_i,u_i\rangle$, it follows that $\|\mu_i\|_2\ge\bigl[(1-\eta)\cos\rho-\eta\bigr]_+$, proving \eqref{eq:mean_norm_lower_bound}.
\end{proof}

\vspace{10pt}

\section{Spherical-Code Bounds and Random-Code Separation}
\label{sec:supp_numerical}

The finite-dimensional spherical-code bounds, asymptotic lower-bound exponents, and uniform-center separation probability are evaluated under the spherical model. These calculations are separate from the DigiFace image-level evaluation. Figure~\ref{fig:bounds_vs_D} evaluates the finite-dimensional bounds in \eqref{eq:packing_sandwich_main} for $(\tau,\rho)=(0.80,8^\circ)$, with $\psi=\psi_\tau(\rho)$.

\begin{figure}[t]
    \centering
    \includegraphics[width=0.84\columnwidth]{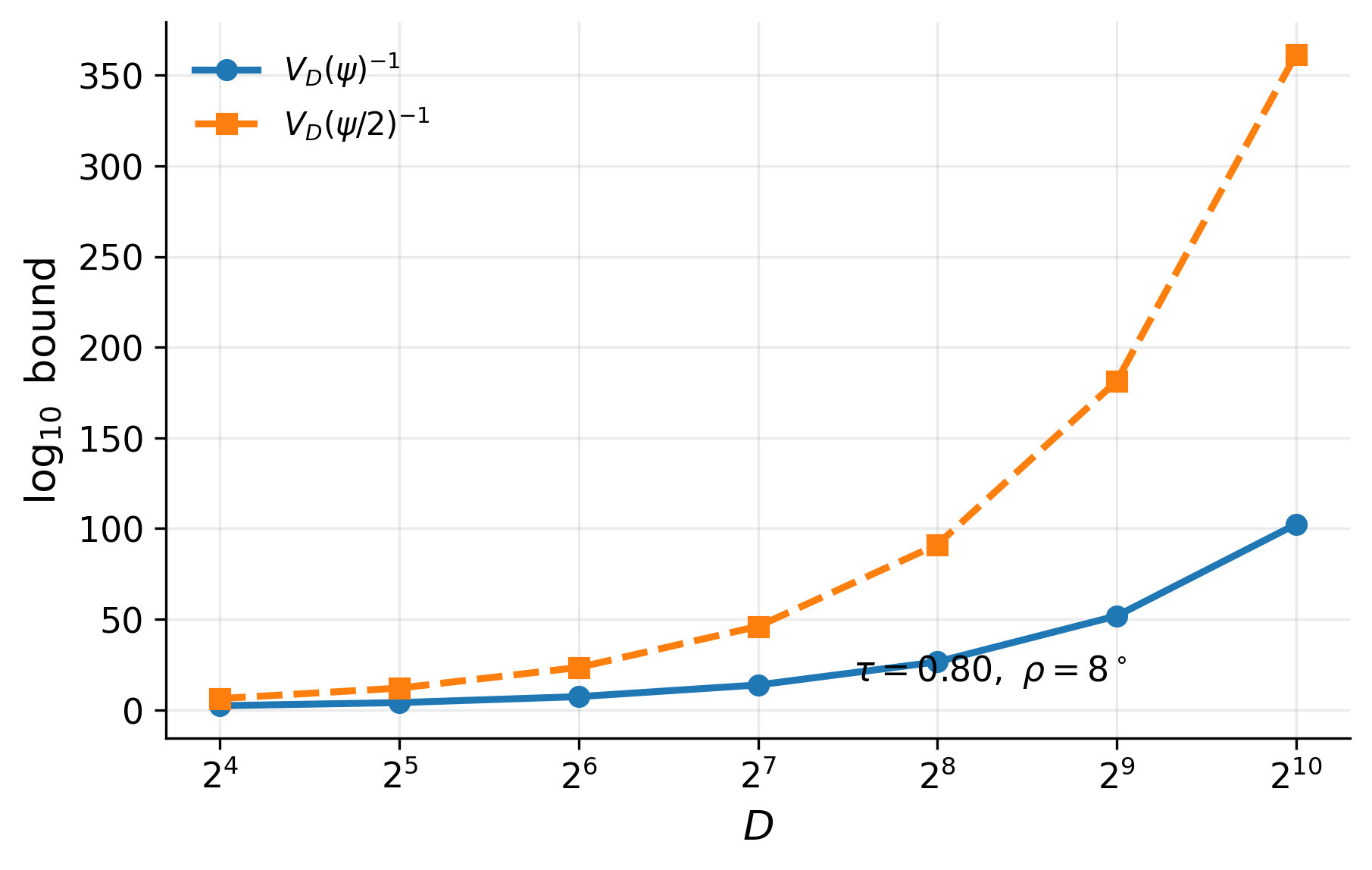}
    \vspace{-5pt}
    \caption{Spherical-code bounds for $(\tau,\rho)=(0.80,8^\circ)$.}
    \label{fig:bounds_vs_D}
\end{figure}

For $\psi_\tau(\rho)<\pi/2$, define the covering-based spherical-code lower-bound exponent
\begin{equation}
R_{\mathrm{sph}}^{\mathrm{LB}}(\tau,\rho) \triangleq -\log \left( \sin\bigl(\psi_\tau(\rho)\bigr) \right).
\end{equation}
If $2\rho<\arccos(\tau)$ and full $(\rho,\eta)$-angular expressivity holds, Theorem~\ref{thm:asymptotic_lower_bound_main} gives the same expression as a lower-bound exponent for the restricted fixed-code capacity. Under the hypotheses of Theorem~\ref{thm:asymptotic_random_capacity_main}, the random-code lower-bound exponent is
\begin{equation}
R_{\mathrm{rnd}}^{\mathrm{LB}}(\tau,\rho) \triangleq
\frac{1}{2} R_{\mathrm{sph}}^{\mathrm{LB}}(\tau,\rho).
\end{equation}

\begin{figure}[t]
    \centering
    \includegraphics[width=0.83\columnwidth]{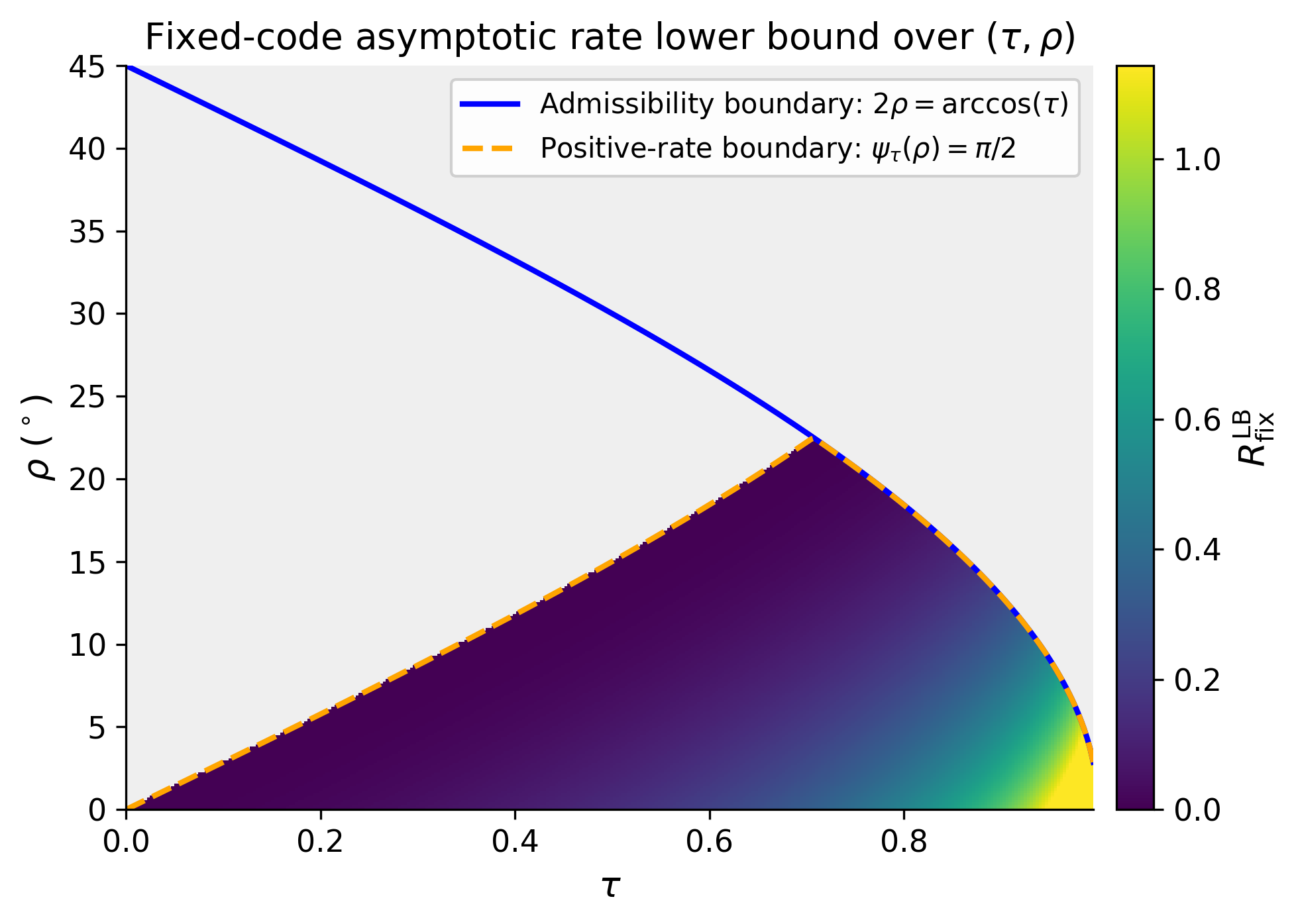}
    \caption{Spherical-code lower-bound exponent.}
    \label{fig:supp_fixed_rate_landscape}
\end{figure}

\begin{figure}[t]
    \centering
    \includegraphics[width=0.85\columnwidth]{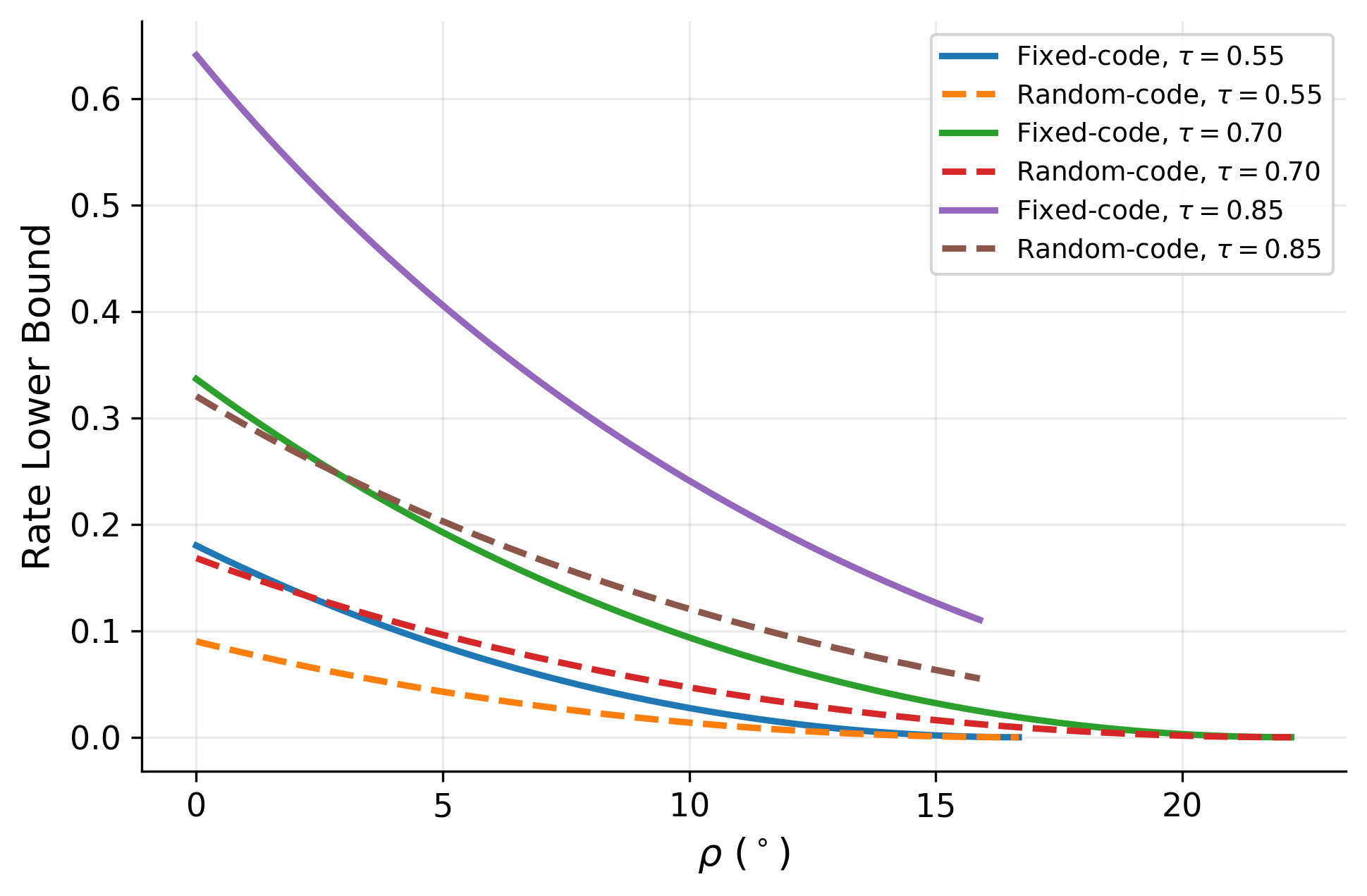}
    \caption{Spherical-code and random-code lower-bound exponents.}
    \label{fig:supp_fixed_vs_random_rate}
\end{figure}

For the finite-dimensional random-code calculation, let $U_1,\ldots,U_M$ be i.i.d.\ uniform on $\mathbb S^{D-1}$ and define
\begin{equation}
P_{\mathrm{sep}}^{(D)}(M;\tau,\rho) \triangleq \Pr \left[
\angle(U_i,U_j)\ge\psi_\tau(\rho), \, \forall\, i\ne j \right].
\end{equation}
The union bound gives
\begin{equation}
P_{\mathrm{sep}}^{(D)}(M;\tau,\rho) \ge 1-\binom{M}{2}V_D \left(\psi_\tau(\rho)\right).
\end{equation}
Figure~\ref{fig:supp_random_sep_success} compares Monte Carlo estimates of $P_{\mathrm{sep}}^{(D)}(M;\tau,\rho)$ with this lower bound.

\begin{figure*}[t]
    \centering
    \includegraphics[
        width=0.85\linewidth,
        trim={0 0 0 1cm},
        clip
    ]{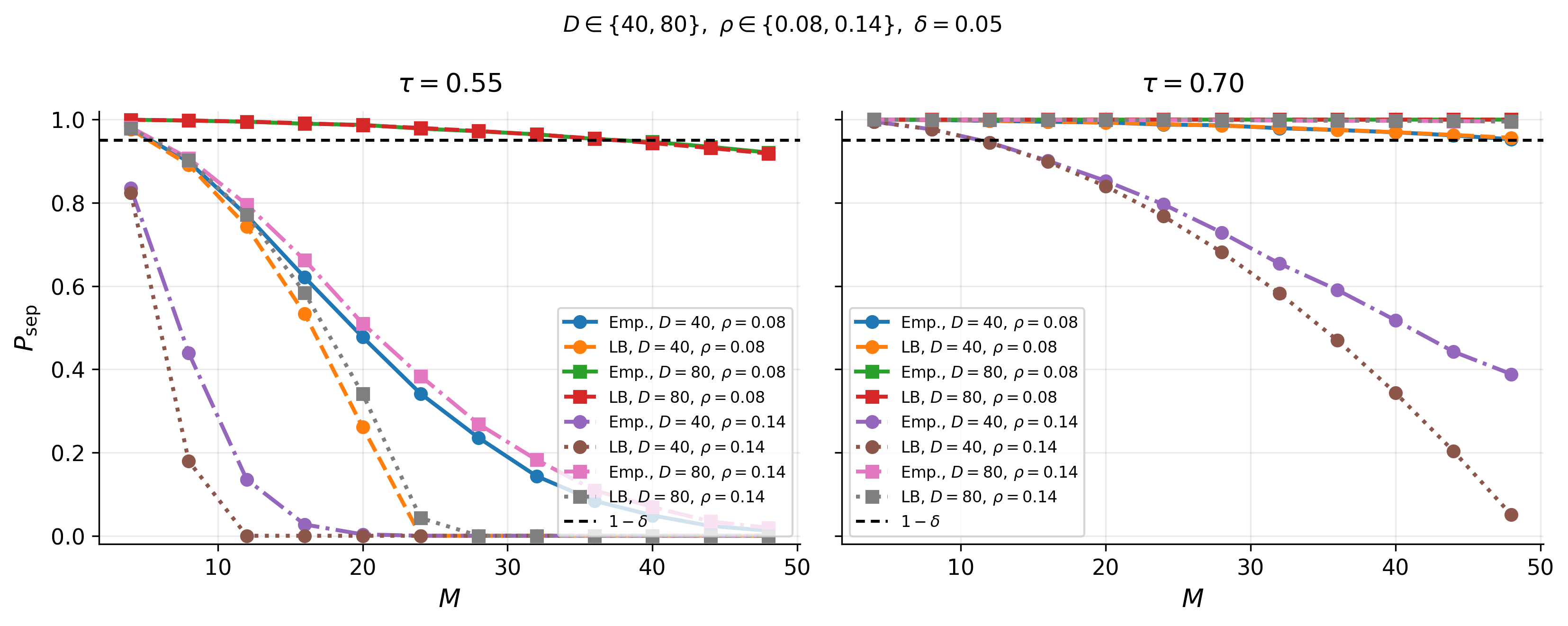}
    \caption{Uniform-center separation probability.}
    \label{fig:supp_random_sep_success}
\end{figure*}

\vspace{20pt}

\section{Image-Level Evaluation Protocol}
\label{sec:supp_real_pipeline_protocol}

In our evaluation, we use released DigiFace-1M repeated-view images under fixed recognizers and a specified calibration protocol. The multiple images available for each generator-defined identity permit evaluation of the empirical genuine acceptance rates and pairwise impostor non-match rates defined in Section~\ref{subsec:finite_sample_evaluation}. The DigiFace generator is not retrained or rerun, and no additional face images or generator-defined identities are produced. The cardinalities are therefore conditional on the fixed calibration and evaluation image sets, the recognizer, and the calibration protocol.

\subsection{Identity Split and Threshold Calibration}

Let $\mathcal I_{\rm cal}$ and $\mathcal I_{\rm eval}$ denote disjoint sets of calibration and evaluation identities. The calibration identities are used only to determine the verification threshold, whereas all sample-restricted evaluation quantities are computed on $\mathcal I_{\rm eval}$. The primary in-domain protocol uses held-out calibration identities from the same selected DigiFace identity pool as the evaluation identities.

In our implementation, we construct the identity split by applying a seeded permutation to a sorted identity list and use the same split for every recognizer. For a calibration fraction $f\in(0,1)$ and $N$ eligible identities, we assign exactly $\lceil fN\rceil$ identities to calibration, subject to at least two identities remaining in each set. We hold the calibration and evaluation image sets fixed across recognizers and verify that the corresponding source-image files are identical.

An LFW calibration set may instead be used to evaluate threshold transfer. A threshold satisfying the target empirical calibration FMR constraint on the LFW calibration set need not satisfy the same empirical FMR constraint on the DigiFace evaluation sample. Such results are threshold-transfer quantities rather than in-domain operating-point quantities.

For recognizer $m$, let $s_m(x,x')$ denote the cosine score after the specified preprocessing. Given the $L_m$ calibration impostor scores $s_{m,1},\ldots,s_{m,L_m}$, define
\begin{equation}
\widehat{\mathrm{FMR}}_m(t) \triangleq
\frac{1}{L_m} \sum_{\ell=1}^{L_m} \mathbf 1\{s_{m,\ell}>t\}.
\label{eq:supp_empirical_calibration_fmr}
\end{equation}
For target empirical calibration FMR $\alpha$, the verification threshold is
\begin{equation}
\tau_{m,\alpha} \triangleq \inf \left\{ t\in[0,1) \,\middle|\,
\widehat{\mathrm{FMR}}_m(t)\le\alpha \right\},
\label{eq:supp_empirical_threshold_rule}
\end{equation}
provided the displayed set is nonempty. The verification rule is
\begin{equation}
\text{match}
\quad\Longleftrightarrow\quad s_m(x,x')>\tau_{m,\alpha}.
\end{equation}
Scores equal to $\tau_{m,\alpha}$ are therefore classified as non-matches under the threshold convention adopted in the main manuscript.

Each recognizer is calibrated separately at the same target $\alpha$. A common numerical cosine threshold is not imposed across recognizers. The calibration target $\alpha$ and the admissibility tolerance $\varepsilon_{\rm out}$ have distinct roles. The former determines the verification threshold from the calibration scores, whereas the latter upper-bounds the empirical impostor match rate for each distinct
evaluation identity pair at that threshold. No equality between $\alpha$ and $\varepsilon_{\rm out}$ is assumed.
If no $t\in[0,1)$ satisfies $\widehat{\mathrm{FMR}}_m(t)\le\alpha$, calibration at the requested target is declared infeasible. We use the strict comparison $s_m>\tau_{m,\alpha}$, consistent with the stated verification rule.

Calibration impostor scores are not mutually independent because different image pairs can share an image. We use $\widehat{\mathrm{FMR}}_m(\tau_{m,\alpha})$ as a descriptive empirical calibration rate and do not assign it a binomial confidence interpretation. Statistical inference for this quantity must account for dependence among the calibration scores.

\subsection{Evaluation Quantities}

For every evaluation identity, we embed and normalize all retained repeated views. Equations~\eqref{eq:empirical_genuine_rate}--\eqref{eq:empirical_impostor_rate} give the complete-pair genuine acceptance and impostor non-match rates. If the repeated views are independent draws from the nuisance distribution, the genuine rate is a U-statistic of order two and the impostor rate is a two-sample U-statistic of order $(1,1)$. For a predetermined render grid, the same quantities are finite-grid descriptive rates. We state the applicable interpretation for each experiment.

After removing identities that fail the empirical genuine constraint, we place an undirected edge between identities $i$ and $j$ when their empirical impostor non-match rate is at least $1-\varepsilon_{\rm out}$. The sample-restricted cardinality is the maximum-clique cardinality of the resulting compatibility graph.  In our implementation, we decompose the graph into connected components and apply exact Bron--Kerbosch search with pivoting~\cite{bron1973algorithm} subject to the configured exact-search limit. A complete component is resolved exactly regardless of its size. If the global maximum-clique cardinality is not established exactly, we return the largest clique found as a certified lower bound.

Fix an empirical coverage-tail parameter $\eta\in[0,1)$. For the same embeddings, let $a_{i,(1)}\le\cdots\le a_{i,(K_i)}$ denote the ordered angles
\[
a_{i,k}=\angle(E_i^{(k)},\widehat u_i)
\]
whenever the normalized empirical mean $\widehat u_i$ is defined. The empirical radius is the smallest observed angular radius containing at least a fraction $1-\eta$ of the evaluated views,
\begin{equation}
\widehat\rho_i = a_{i,(\lceil(1-\eta)K_i\rceil)}.
\label{eq:supp_empirical_radius_order_stat}
\end{equation}
These quantities are used only in the finite-sample geometric check. We use $\zeta=10^{-12}$ in \eqref{eq:empirical_mean_direction_main}; hence $\widehat u_i$ and $\widehat\rho_i$ are defined only when $\|\widehat m_i\|_2>\zeta$.

The finite-sample geometric check requires
\begin{align}
2\widehat\rho_\ell &<\arccos(\tau),
\qquad \ell\in\{i,j\},
\\
\angle(\widehat u_i,\widehat u_j) &\ge \arccos(\tau)+\widehat\rho_i+\widehat\rho_j.
\end{align}
Satisfaction of these inequalities does not by itself establish the population centeredness conditions or population-level admissibility without an additional statistical coverage argument. It also does not by itself imply the empirical admissibility constraints defining $\widehat C_{\mathcal S}$.

\subsection{Recognizer Comparisons}

To study recognizer dependence, we use the same calibration and evaluation images for every recognizer and calibrate each recognizer separately at a common target empirical calibration FMR. For this comparison, all recognizers receive the same retained aligned face crops, followed by model-specific input normalization. The sample-restricted cardinality is computed under each fixed recognizer and its specified calibration protocol. The TFace/ADMIS comparison measures differences between the selected recognition pipelines. Because the corresponding training runs are not controlled to differ only in training-data source, the comparison does not isolate a causal effect of training-data source.

\subsection{DigiFace-1M Evaluation}

We use the released DigiFace-1M 72-view images~\cite{bae2023digiface1m} without rerunning the DigiFace generator. From the P1 block corresponding to identities 0--1999, a deterministic selection with seed 20260901 chooses 250 identities and 24 images per identity, yielding 6,000 source images. Before recognizer scoring, we apply YuNet~\cite{wu2023yunet} to the selected images and evaluate detector score thresholds $0.90,0.85,0.80,0.75,0.70$. The threshold $0.80$ is the largest value in this grid for which every selected identity retains at least 12 of its 24 views. The retained-image counts at thresholds $0.90$, $0.85$, and $0.80$ are 4,481, 5,736, and 5,868, respectively. We therefore use the detector score threshold $0.80$. This criterion controls preprocessing coverage and does not certify landmark accuracy. We align every retained detection to $112\times112$ from the five YuNet landmarks~\cite{wu2023yunet} using OpenCV \texttt{FaceRecognizerSF::alignCrop}\footnote{OpenCV \texttt{FaceRecognizerSF} implementation, \url{https://github.com/opencv/opencv/blob/4.x/modules/objdetect/src/face_recognize.cpp}.}, consistent with the five-point $112\times112$ preprocessing described in~\cite{deng2019arcface}. The common preprocessing stage retains 5,868 images from all 250 identities, with 13--24 retained views per identity. The remaining 132 selected images do not produce retained aligned crops. We use the same 5,868 aligned crops for all three recognizers.

The evaluated recognizers are OpenCV SFace~\cite{zhong2021sface}, which produces 128-dimensional embeddings, a public TFace IR-50 checkpoint\footnote{TFace recognition models, \url{https://github.com/Tencent/TFace/tree/master/recognition}.}, which produces 512-dimensional embeddings, and the public ADMIS Syn-30k IR-50/ArcFace checkpoint~\cite{deng2019arcface,deandres2024frcsyn}\footnote{ADMIS pretrained models, \url{https://github.com/zzzweakman/CVPR24_FRCSyn_ADMIS}.}, which also produces 512-dimensional embeddings. The TFace release provides an IR-50 model and separately lists an IR-50/ArcFace/MS1Mv2 benchmark configuration, but does not establish that the released IR-50 checkpoint used in this evaluation is that exact benchmark model. We therefore describe it only as the TFace IR-50 checkpoint. The ADMIS release identifies the evaluated model as an IR-50 model trained on Syn-30k with ArcFace. The TFace and ADMIS results are therefore treated as a comparison of two fixed IR-50 recognition pipelines, not as a controlled experiment in which training-data source is the only changed variable.

For the primary in-domain operating point, we use seed 20260901 to assign 50 DigiFace identities to calibration and 200 identities to evaluation, with the same identity split for all recognizers. After the common preprocessing stage, the calibration and evaluation sets contain 1,179 and 4,689 images, respectively. The two sets contain disjoint identities and images. Complete calibration impostor scoring gives an attained
empirical FMR of approximately $9.9984\times10^{-4}$ at target $10^{-3}$ for every recognizer. The calibrated thresholds are $0.522704$ for SFace, $0.496364$ for ADMIS Syn-30k IR-50, and $0.398517$ for TFace IR-50.

With  $\varepsilon_{\rm in}=\varepsilon_{\rm out}=0.05$, the pairs $(N_{\rm G},\widehat C_{\mathcal S})$ are $(2,2)$ for SFace, $(9,9)$ for ADMIS Syn-30k IR-50, and $(70,58)$ for TFace IR-50. All three maximum-clique cardinalities are exact. We use $\eta=0.05$ only for the finite-sample geometric check based on empirical mean directions and radii, so it does not enter the operational compatibility graph. For TFace, a maximum compatible subset therefore contains $58/70\approx82.9\%$ of the identities satisfying the empirical genuine constraint on the primary split. For SFace and ADMIS, these identities form complete compatibility graphs, so $\widehat C_{\mathcal S}=N_{\rm G}$ on this split.

\paragraph{Threshold Transfer from LFW}

We separately evaluate threshold transfer using an aligned LFW~\cite{huang2008labeled} calibration subset\footnote{Aligned LFW data release used in the experiment, \url{https://github.com/haibo-qiu/SynFace}.} containing 250 identities with two images each and 124,500 cross-identity impostor scores. At target empirical LFW calibration FMR $10^{-3}$, the calibrated thresholds are $0.351269$ for SFace, $0.350165$ for ADMIS Syn-30k IR-50, and $0.236878$ for TFace IR-50. We apply these thresholds without recalibration to the selected 250-identity DigiFace sample. The resulting aggregate complete-pair empirical DigiFace impostor FMRs are approximately $7.33\%$, $2.55\%$, and $9.06\%$ for SFace, ADMIS Syn-30k IR-50, and TFace IR-50, respectively.

At the transferred thresholds, $(N_{\rm G},\widehat C_{\mathcal S})$ is $(87,12)$ for SFace and $(90,23)$ for ADMIS, with both maximum-clique cardinalities exact. For TFace, $N_{\rm G}=234$ and $\widehat C_{\mathcal S}\ge18$, where 18 is a certified lower bound. The LFW-calibrated thresholds therefore yield empirical DigiFace FMRs above the target value of $10^{-3}$ on the selected 250-identity sample. These cardinalities are threshold-transfer quantities and are not interpreted as DigiFace results at an in-domain empirical FMR of $10^{-3}$. This analysis quantifies sensitivity to threshold calibration and is not used as the primary operating point.

\subsection{Effective-Dimension Sensitivity}

For the effective-dimension sensitivity analysis, we generate one seeded random orthonormal basis and define nested subspaces using its first $d$ basis vectors. We project the calibration and evaluation embeddings onto the same subspace at each dimension, renormalize the projected embeddings, recalibrate the verification threshold separately at the same target empirical calibration FMR, and recompute the sample-restricted
quantities. This procedure evaluates dimensional-compression sensitivity of one trained representation. It does not test the $D\to\infty$ result, which concerns a specified family of pipelines with increasing native embedding dimension.

For TFace IR-50, seed 20260901 defines the orthonormal basis and the nested dimensions $d\in\{16,32,64,128,256,512\}$. We use the same 50 calibration identities and 200 evaluation identities at every $d$, with separate threshold recalibration at target empirical calibration FMR $10^{-3}$ after projection and renormalization. The corresponding $N_{{\rm G},d}$ values are 0, 0, 1, 12, 34, and 70, and the exact maximum-clique cardinalities are 0, 0, 1, 12, 30, and 58. All maximum-clique cardinalities in this analysis are exact. As a separate SFace sensitivity analysis under the same in-domain split and calibration target, the exact maximum-clique cardinalities at $d=16,32,64,128$ are 0, 0, 0, and 2, respectively. Table~\ref{tab:effective_dimension_sensitivity} gives the complete TFace results.

\subsection{Reproducibility Records}

For every evaluation, we record the calibration and evaluation identity sets, retained images, recognizer configuration, calibrated threshold, target and attained empirical calibration FMR, $(\varepsilon_{\rm in},\varepsilon_{\rm out},\eta)$, per-identity genuine acceptance rates, pairwise impostor non-match rates, and whether the maximum-clique cardinality is exact or a certified lower bound. We also record the empirical mean directions and radii used in the separate geometric check.

We verify that the retained image sets and identity split used for each cross-recognizer comparison are identical across recognizers and that the calibration and evaluation sets contain disjoint identities and source images. The retained embeddings and recorded evaluation parameters are sufficient to reconstruct the compatibility graph and recompute the sample-restricted cardinality without rerunning the recognizer.

For the finite-sample geometric check, the pass fraction is computed over all unordered pairs of evaluation identities. A pair is counted as passing only when both empirical mean directions are defined, both identities satisfy
\begin{equation}
2\widehat\rho_\ell<\arccos(\tau), \qquad \ell\in\{i,j\},
\end{equation}
and the pair satisfies
\begin{equation}
\angle(\widehat u_i,\widehat u_j) \ge \arccos(\tau)+\widehat\rho_i+\widehat\rho_j.
\end{equation}
Pairs failing any of these conditions are counted as non-passing. We separately record whether every unordered pair in the returned subset satisfies the same geometric check. Neither quantity is interpreted as a population-level admissibility certificate.

\section{Geometric Converse under Full-Cap Support}
\label{sec:full_cap_converse}

Theorem~\ref{thm:sufficient_admissibility_main} gives a sufficient admissibility condition under the centered model but no converse. To derive necessary geometric conditions, we strengthen the model by requiring each identity-conditional embedding distribution to have support equal to a radius-$\rho$ cap. We restrict to $\rho\in[0,\pi/2)$.

\begin{definition}[$\rho$-full-cap-supported identity-conditional embedding distribution]
\label{def:full_cap_supported}
An identity-conditional embedding distribution $P_i$ is \textit{$\rho$-full-cap-supported at}
$u_i\in\mathbb S^{D-1}$ if
\begin{equation}
\operatorname{supp}(P_i) = \mathrm{Cap}(u_i,\rho),
\label{eq:full_cap_support_def}
\end{equation}
where support is relative to $\mathbb S^{D-1}$. It is called \textit{$\rho$-full-cap-supported} if such a $u_i$ exists.
\end{definition}

\begin{theorem}[Full-cap zero-error conditions]
\label{thm:full_cap_iff}
Fix $\rho\in[0,\pi/2)$ and $\tau\in[0,1)$, and suppose $\{P_i\}_{i=1}^M$ is a family such that $P_i$ is $\rho$-full-cap-supported at $u_i\in\mathbb S^{D-1}$ for every $i\in[M]$.
\begin{enumerate}
\item If the family is $(\tau,0,0)$-admissible, then
\begin{align}
2\rho &\le \arccos(\tau),
\label{eq:full_cap_intra_iff}\\
\angle(u_i,u_j) &\ge \arccos(\tau)+2\rho, \qquad i\ne j.
\label{eq:full_cap_inter_iff}
\end{align}
\item If
\begin{equation}
2\rho<\arccos(\tau)
\label{eq:full_cap_intra_suff}
\end{equation}
and \eqref{eq:full_cap_inter_iff} holds, then the family is $(\tau,0,0)$-admissible.
\end{enumerate}
At the boundary $2\rho=\arccos(\tau)$, every within-cap pair has cosine score at least $\tau$, but zero-error genuine acceptance additionally requires
\begin{equation}
(P_i\otimes P_i) \left( \left\{ (e,e')\in \mathbb S^{D-1}\times\mathbb S^{D-1} \,\middle|\,
\langle e,e'\rangle=\tau \right\} \right) =0
\end{equation}
for every $i\in[M]$. Thus support alone does not determine the zero-error genuine condition at the boundary.
\end{theorem}

\begin{proof}
For necessity of \eqref{eq:full_cap_intra_iff}, suppose $2\rho>\arccos(\tau)$. The cap $\mathrm{Cap}(u_i,\rho)$ contains two points whose angular separation is $2\rho$ and whose cosine score is therefore strictly smaller than $\tau$. By continuity of the cosine score, there exist relative neighborhoods of these points within the cap such that every pair formed by one point from each neighborhood has score strictly smaller than $\tau$. Full support assigns positive $P_i$-probability to each neighborhood. Since the two genuine samples are independent, the event that they lie in the respective neighborhoods has positive probability. The genuine rejection probability is therefore positive, contradicting $(\tau,0,0)$-admissibility. Hence $2\rho\le\arccos(\tau)$.

For necessity of \eqref{eq:full_cap_inter_iff}, suppose for some $i\ne j$ that
\begin{equation}
\angle(u_i,u_j) < \arccos(\tau)+2\rho.
\end{equation}
Then $\mathrm{Cap}(u_i,\rho)$ and $\mathrm{Cap}(u_j,\rho)$ contain points whose angular separation is strictly smaller than $\arccos(\tau)$. Their cosine score is therefore strictly larger than $\tau$. By continuity, there exist relative neighborhoods of these points within the respective caps such that every corresponding cross-cap pair has score strictly larger than $\tau$. Full support assigns positive probability to both neighborhoods. Independence of the cross-identity samples therefore gives a positive-probability impostor match event, contradicting $(\tau,0,0)$-admissibility. Hence \eqref{eq:full_cap_inter_iff} is necessary.

For sufficiency, \eqref{eq:full_cap_intra_suff} implies that every pair of points in the same radius-$\rho$ cap has angular separation at most $2\rho<\arccos(\tau)$ and therefore cosine score strictly larger than $\tau$. Thus the genuine constraint holds with probability one. Under \eqref{eq:full_cap_inter_iff}, every pair of points from distinct caps has angular separation at least
\[
\angle(u_i,u_j)-2\rho \ge \arccos(\tau)
\]
and therefore cosine score at most $\tau$. Thus the impostor non-match constraint also holds with probability one. The family is consequently $(\tau,0,0)$-admissible.
\end{proof}

\begin{definition}[$\rho$-full-cap restricted identity capacity]
\label{def:full_cap_restricted_capacity}
Define
\begin{multline}
C_{D,\mathrm{cap}}^{(\rho)}(\tau;g,\phi) \triangleq \sup\Bigl( \{0\}\cup
\Bigl\{ M\in\mathbb N  \,\Big|\, \\
\exists\,c_1,\ldots,c_M\in\mathcal Z
\text{ pairwise distinct},\\
\{P_{c_i}\}_{i=1}^M
\text{ is }(\tau,0,0)\text{-admissible}, \\
\text{each }P_{c_i}\text{ is }\rho\text{-full-cap-supported} \Bigr\} \Bigr).
\label{eq:full_cap_capacity_def}
\end{multline}
\end{definition}

\begin{corollary}[Restricted capacity under full-cap support]
\label{cor:full_cap_packing_upper}
Let $\rho\in[0,\pi/2)$. If $2\rho>\arccos(\tau)$, then
\begin{equation}
C_{D,\mathrm{cap}}^{(\rho)}(\tau;g,\phi) = 0.
\label{eq:full_cap_zero_capacity}
\end{equation}
If instead $2\rho\le\arccos(\tau)$, then $\psi_\tau(\rho)\le\pi$ and
\begin{equation}
C_{D,\mathrm{cap}}^{(\rho)}(\tau;g,\phi) \le A_D\bigl(\psi_\tau(\rho)\bigr).
\label{eq:full_cap_packing_upper}
\end{equation}
\end{corollary}

\begin{proof}
If $2\rho>\arccos(\tau)$, the necessary condition \eqref{eq:full_cap_intra_iff} rules out even a one-identity
$(\tau,0,0)$-admissible family. Hence \eqref{eq:full_cap_zero_capacity} follows. Suppose instead that $2\rho\le\arccos(\tau)$. Since $\tau\in[0,1)$,
\[
\psi_\tau(\rho) = \arccos(\tau)+2\rho \le 2\arccos(\tau) \le \pi.
\]
By Theorem~\ref{thm:full_cap_iff}, the centers of every admissible family satisfy
\[
\angle(u_i,u_j) \ge \psi_\tau(\rho), \qquad i\ne j.
\]
They therefore form a spherical code with minimum angle at least $\psi_\tau(\rho)$. Its cardinality cannot exceed $A_D(\psi_\tau(\rho))$, which proves \eqref{eq:full_cap_packing_upper}.
\end{proof}

\begin{definition}[Full $\rho$-cap angular expressivity]
\label{def:full_cap_angular_expressivity}
The pair $(g,\phi)$ is \textit{fully $\rho$-cap angularly expressive} if, for every finite set of pairwise-distinct directions $u_1,\ldots,u_M\in\mathbb S^{D-1}$, there exist pairwise-distinct latent codes $c_1,\ldots,c_M\in\mathcal Z$ such that
\begin{equation}
\operatorname{supp}(P_{c_i}) = \mathrm{Cap}(u_i,\rho), \qquad i\in[M].
\end{equation}
\end{definition}

\vspace{5pt}

\begin{corollary}[Exact full-cap restricted capacity]
\label{cor:full_cap_exact_reduction}
If $(g,\phi)$ is fully $\rho$-cap angularly expressive and $2\rho<\arccos(\tau)$, then
\begin{equation}
C_{D,\mathrm{cap}}^{(\rho)}(\tau;g,\phi) = A_D\bigl(\psi_\tau(\rho)\bigr).
\label{eq:full_cap_exact_reduction}
\end{equation}
\end{corollary}

\begin{proof}
The upper bound follows from Corollary~\ref{cor:full_cap_packing_upper}. Let $u_1,\ldots,u_M\in\mathbb S^{D-1}$ be a spherical code with $M=A_D\bigl(\psi_\tau(\rho)\bigr)$ and minimum pairwise angle at least $\psi_\tau(\rho)$. Full $\rho$-cap angular expressivity provides pairwise-distinct latent codes $c_1,\ldots,c_M$ such that
\[
\operatorname{supp}(P_{c_i}) = \mathrm{Cap}(u_i,\rho), \qquad i\in[M].
\]
Since $2\rho<\arccos(\tau)$, the sufficient part of Theorem~\ref{thm:full_cap_iff} implies that the induced family is $(\tau,0,0)$-admissible. Therefore
\[
C_{D,\mathrm{cap}}^{(\rho)}(\tau;g,\phi) \ge A_D\bigl(\psi_\tau(\rho)\bigr).
\]
Combining this inequality with the upper bound proves \eqref{eq:full_cap_exact_reduction}.
\end{proof}

\renewcommand{\refname}{Appendix References}
\putbib
\end{bibunit}

\end{document}